\documentclass[11pt,a4paper,reqno]{amsart}
\usepackage[margin=1in]{geometry}
\usepackage[citestyle=numeric,style=numeric,maxbibnames=99,natbib=true,sortcites]{biblatex}
\usepackage[english]{babel}
\usepackage{amsmath}
\usepackage{amsthm}
\usepackage{enumerate}
\usepackage[hidelinks]{hyperref}
\usepackage{pgfplots}
\usepackage{booktabs}
\usepackage{xcolor}
\usepackage{algorithm}
\usepackage[noend]{algpseudocode}
\usepackage{graphicx}
\usepackage{enumitem}
\usepackage[capitalise]{cleveref}
\usepackage{etoolbox}
\usepackage{multirow}
\usepackage[textsize=tiny]{todonotes}
\usepackage{thmtools, thm-restate}

\makeatletter
\patchcmd{\NAT@test}{\else \NAT@nm}{\else \NAT@hyper@{\NAT@nm}}{}{}
\makeatother

\usepackage{tikz}
\usepackage{tikz-network}
\usetikzlibrary{patterns}
\usetikzlibrary{patterns.meta}
\usetikzlibrary{decorations.pathreplacing,calligraphy}
\usetikzlibrary{decorations.pathmorphing}
\usetikzlibrary{backgrounds}
\usetikzlibrary{arrows}
\usetikzlibrary{plotmarks}
\usetikzlibrary{fadings}

\colorlet{lightblue}{cyan!40!white}
\colorlet{lightgreen}{lime!60!white}
\colorlet{lightorange}{orange!50!white}
\colorlet{lightviolet}{violet!40!white}
\colorlet{dark gray}{black}

\algrenewcommand\algorithmicrequire{\textbf{Input:}}
\algrenewcommand\algorithmicensure{\textbf{Output:}}
\algnewcommand\algorithmicforeach{\textbf{for each}}
\algdef{S}[FOR]{ForEach}[1]{\algorithmicforeach\ #1\ \algorithmicdo}

\makeatletter
\tikzset{
        hatch distance/.store in=\hatchdistance,
        hatch distance=5pt,
        hatch thickness/.store in=\hatchthickness,
        hatch thickness=5pt
        }
\pgfdeclarepatternformonly[\hatchdistance,\hatchthickness]{north east hatch}
    {\pgfqpoint{-1pt}{-1pt}}
    {\pgfqpoint{\hatchdistance}{\hatchdistance}}
    {\pgfpoint{\hatchdistance-1pt}{\hatchdistance-1pt}}%
    {
        \pgfsetcolor{\tikz@pattern@color}
        \pgfsetlinewidth{\hatchthickness}
        \pgfpathmoveto{\pgfqpoint{0pt}{0pt}}
        \pgfpathlineto{\pgfqpoint{\hatchdistance}{\hatchdistance}}
        \pgfusepath{stroke}
    }
\pgfdeclarepatternformonly[\hatchdistance,\hatchthickness]{north west hatch}
    {\pgfqpoint{-0pt}{-0pt}}
    {\pgfqpoint{\hatchdistance}{\hatchdistance}}
    {\pgfpoint{\hatchdistance}{\hatchdistance}}%
    {
        \pgfsetcolor{\tikz@pattern@color}
        \pgfsetlinewidth{\hatchthickness}
        \pgfpathmoveto{\pgfqpoint{0pt}{\hatchdistance}}
        \pgfpathlineto{\pgfqpoint{\hatchdistance}{0pt}}
        \pgfpathmoveto{\pgfqpoint{.2pt}{-.2pt}}
        \pgfpathlineto{\pgfqpoint{-.2pt}{.2pt}}
        \pgfpathmoveto{\pgfqpoint{\hatchdistance+0.2pt}{\hatchdistance-0.2pt}}
        \pgfpathlineto{\pgfqpoint{\hatchdistance-0.2pt}{\hatchdistance+0.2pt}}
        \pgfusepath{stroke}
    }
\pgfdeclarepatternformonly{south east lines}{\pgfqpoint{-0pt}{-0pt}}{\pgfqpoint{3pt}{3pt}}{\pgfqpoint{3pt}{3pt}}{
    \pgfsetlinewidth{0.4pt}
    \pgfpathmoveto{\pgfqpoint{0pt}{3pt}}
    \pgfpathlineto{\pgfqpoint{3pt}{0pt}}
    \pgfpathmoveto{\pgfqpoint{.2pt}{-.2pt}}
    \pgfpathlineto{\pgfqpoint{-.2pt}{.2pt}}
    \pgfpathmoveto{\pgfqpoint{3.2pt}{2.8pt}}
    \pgfpathlineto{\pgfqpoint{2.8pt}{3.2pt}}
    \pgfusepath{stroke}}
\makeatother

\theoremstyle{plain}
\newtheorem{theorem}{Theorem}[section]
\newtheorem{lemma}[theorem]{Lemma}
\newtheorem{proposition}[theorem]{Proposition}
\newtheorem{corollary}[theorem]{Corollary}
\newtheorem*{theorem*}{Claim}

\theoremstyle{definition}

\newcommand{\R}{\mathbb{R}}
\newcommand{\N}{\mathbb{N}}

\newcommand{\calI}{\mathcal{I}}

\newcommand{\Exp}{\mathbb{E}}
\newcommand{\E}[1]{\Exp[ #1 ]}

\renewcommand{\P}[1]{\mathbb{P}[ #1 ]}

\newcommand{\OPT}{\normalfont\textsc{Opt}}

\newcommand{\scM}{\mathrm{M}}
\newcommand{\nalgspgreedy}{\normalfont\textsc{Greedy}}
\newcommand{\algspgreedy}{\normalfont\textsc{G}}
\newcommand{\nalgmixspgreedy}{\textsc{SingleGreedy}}
\newcommand{\algmixspgreedy}{\textsc{G}_{\normalfont\text{sgl}}}
\newcommand{\nalgrandspgreedy}{\textsc{RandomizedGreedy}}
\newcommand{\algrandspgreedy}{\textsc{G}_{\normalfont\text{rnd}}}
\newcommand{\nalgud}{\normalfont\textsc{FitTwo}}
\newcommand{\algud}{\normalfont\textsc{F}}
\newcommand{\nalgrspgreedy}{\normalfont\textsc{RestrictedGreedy}}
\newcommand{\algrspgreedy}{\textsc{G}_{\normalfont\text{res}}}
\newcommand{\nalgmixud}{\textsc{RandomizedFit}}
\newcommand{\algmixud}{\normalfont\textsc{F}_{\normalfont\text{rnd}}}
\newcommand{\nalglarge}{\normalfont\textsc{LargeFit}}
\newcommand{\alglarge}{\normalfont\textsc{F}_{\normalfont\text{lar}}}
\newcommand{\lfrac}[2]{#1/#2}

\vfuzz \hfuzz
\bibliography{WINE2024/selfishknapsack}

\begin{document}

\title{Packing a Knapsack with Items Owned by Strategic Agents}
\author[J.~Cembrano]{Javier Cembrano}
\address[J.~Cembrano]{Department of Algorithms and Complexity, Max Planck Institut für Informatik, Germany}
\email{\href{mailto:jcembran@mpi-inf.mpg.de}{jcembran@mpi-inf.mpg.de}}
\urladdr{\url{https://sites.google.com/view/javier-cembrano/}}

\author[M.~Klimm]{Max Klimm}
\address[M.~Klimm]{Institute for Mathematics, Technische Universität Berlin, Germany}
\email{\href{mailto:klimm@math.tu-berlin.de}{klimm@math.tu-berlin.de}}
\urladdr{\url{https://www3.math.tu-berlin.de/disco/team/klimm/}}

\author[M.~Knaack]{Martin Knaack}
\address[M.~Knaack]{Institute for Mathematics, Technische Universität Berlin, Germany}
\email{\href{mailto:knaack@math.tu-berlin.de}{knaack@math.tu-berlin.de}}
\urladdr{\url{https://www3.math.tu-berlin.de/disco/team/knaack/}}

\begin{abstract}
This paper considers a scenario within the field of mechanism design without money where a mechanism designer is interested in selecting items with maximum total value under a knapsack constraint. The items, however, are controlled by strategic agents who aim to maximize the total value of their items in the knapsack. This is a natural setting, e.g., when agencies select projects for funding, companies select products for sale in their shops, or hospitals schedule MRI scans for the day.
A mechanism governing the packing of the knapsack is strategyproof if no agent can benefit from hiding items controlled by them to the mechanism. We are interested in mechanisms that are strategyproof and $\alpha$-approximate in the sense that they always approximate the maximum value of the knapsack by a factor of $\alpha \in [0,1]$.
First, we give a deterministic mechanism that is $\smash{\frac{1}{3}}$-approximate.
For the special case where all items have unit density, we design a $\smash{\frac{1}{\phi}}$-approximate mechanism where $1/\phi \approx 0.618$ is the inverse of the golden ratio. This result is tight as we show that no deterministic strategyproof mechanism with a better approximation exists.
We further give randomized mechanisms with approximation guarantees of $1/2$ for the general case and $2/3$ for the case of unit densities. For both cases, no strategyproof mechanism can achieve an approximation guarantee better than $1/(5\phi -7)\approx 0.917$.
\end{abstract}

\maketitle

\section{Introduction}

Decision-making in competitive environments often involves the selection of multiple items under a knapsack constraint: 
Funding agencies with a fixed budget select a subset of projects submitted to them for funding, shops stock their shelves with products that they can procure, and hospitals schedule MRI scans for a day.
These situations can be modeled straightforwardly as a knapsack problem with a set of items $E$ where for each item $i \in E$ we are given a size $s_i > 0$ and a value $v_i > 0$. The size of an item describes its consumption of a certain resource such as the total money needed to fund the project, the shelf meters needed to show the product, or the (estimated) length of the MRI scan. The value of an item is an estimator for the benefit to the system designer for selecting the item, e.g., the estimated societal benefit for funding a project, the estimated profits from selling the product, or the estimated benefit from conducting the MRI scan. Assuming that the system designer holds a total capacity of $C > 0$ at their disposal, a natural objective is to maximize the total value of items that can be selected under the knapsack constraint, i.e., to find
\begin{align*}
\OPT(E) = \max \bigl\{ v(E') \;\big\vert\; s(E') \leq C, E' \subseteq E\bigr\},
\end{align*}
where, for a subset $E' \subseteq E$, we set $v(E') = \sum_{i \in E'} v_i$ and $s(E') = \sum_{i \in E'} s_i$.

In many applications, including the ones above, the items in question are not controlled by the system designer but rather by different agents following their own interests.
For illustration, consider the following situation where researchers submit project proposals to a funding agency and the funding agency selects a subset of the submitted proposals for funding.\footnote{A similar example is given in \citet{FeigenbaumJ17}.} Suppose that the funding agency has a total budget of 1. Researcher~$A$ has two proposals requiring a budget of $1/2$ and $2/3$ respectively. Researcher~$B$ has only a single proposal requiring a budget of $1/2$. Assume that for each project the value is equal to the budget spent. If both researchers submit all their proposals to the agency, both projects with a budget of $1/2$ can be accepted, generating a total value of $1$. However, if Researcher~$A$ only submits the proposal with a budget of $2/3$, the optimal way for the agency to spend its budget is to only accept this proposal. Thus, by not submitting one of their proposals, Researcher~$A$ managed to decrease the value of the agency from $1$ to $2/3$ while increasing the total budget of accepted own projects from $1/2$ to $2/3$.
Such strategic behavior of Researcher~$A$ is clearly undesirable from the point of view of the research agency.

As a remedy, \citet{chen2011mechanism} as well as \citet{FeigenbaumJ17} study strategyproof mechanisms for the knapsack problem.
In this setting, which belongs to the growing field of mechanism design without money, the items are owned by a finite set of agents $[n] = \{1,\dots,n\}$ such that each item~$i \in E$ is owned by a single agent~$a(i) \in [n]$, and a single agent may own multiple items $E_a = \{i \in E \mid a(i) = a\}$. A deterministic mechanism is a function~$\scM$ that receives the knapsack problem, i.e., the set of items $E$ with their values and sizes, the information about the ownership of items, and the capacity $C$ as input, and returns a knapsack solution $\scM(E) \subseteq E$ that is feasible, i.e., $s(\scM(E)) \leq C$.
A mechanism is called \emph{strategyproof} if no agent can increase the total value of their items in the knapsack solution by hiding own items to the mechanism, i.e.,
\begin{align*}
v(E_a \cap \scM(E)) \geq v(E_a \cap \scM(E')) \text{ for all $a \in [n]$ and all } E' \subseteq E \text{ with } E \setminus E' \subseteq E_a.
\end{align*}
From a general mechanism design point of view, we want to stress that we only consider the case where agents report a \mbox{(sub-)}set $E_a' \subseteq E_a$ of their items to the mechanism. We do not, however, allow agents to misreport the values or sizes of their items, or to report ``fake'' items that they do not own.
This is a suitable assumption since in the applications above reporting fake items or wrong item sizes can be detected easily and, thus, be punished by the system designer. Also in the applications above, the value of an item can be observed (or is determined) by the system designer. On the other hand, the system designer cannot observe whether the full set of items available to an agent has been reported, thus, making their report vulnerable to the agents' strategic behavior.

Since the capacity of the knapsack as well as the sizes and values of the items are fixed, it is useful to consider the mechanism only as a function of the items that it receives as input.
A mechanism is called $\alpha$-approximate for some $\alpha \in [0,1]$ if the total value of the items packed by the mechanism is always at least an $\alpha$-fraction of the optimal value. 
\citet{FeigenbaumJ17} showed that no deterministic strategyproof mechanism can be $\alpha$-approximate for $\alpha > 1/\phi \approx 0.618$ and that no randomized strategyproof mechanism can be $\alpha$-approximate for $\alpha > 1/(5\phi-7) \approx 0.917$, where $\phi = (1+ \sqrt{5})/2$ is the golden ratio.\footnote{A \emph{randomized} mechanism is a lottery over deterministic mechanism and it is strategyproof (in expectation) if no agent can increase the expected value of their items in the solution by hiding items to the mechanism.}
These results hold even for instances with two agents. On the positive side, \citet{chen2011mechanism} gave a strategyproof randomized mechanism with a constant approximation guarantee. Their analysis yields an approximation factor of $\alpha = 1/576 \approx 0.002$; see \citet{choi2021}.  However, no deterministic mechanism with a constant approximation factor was known.

\subsection{Our Contribution and Techniques}

We give deterministic and randomized mechanisms with improved approximation guarantees. For a comparison of our and previous results, see~\Cref{tab:results}.

We start by designing the first \emph{deterministic} strategyproof mechanism with a constant approximation guarantee for the strategic knapsack problem. Specifically, our mechanism provides an approximation factor of $1/3$ for an arbitrary number of agents; see~\Cref{theo:general-det-UB}. For the special case of two agents, the approximation guarantee can be improved to $1/2$; see~\Cref{cor:det-LB-1n}.

Secondly, we significantly reduce the gap between lower and upper bounds on the approximation guarantee of strategyproof randomized mechanisms. We devise a randomized mechanism with an approximation factor of $1/2$ (see~\Cref{theo:general-rand-UB}) improving on the previous guarantee of $1/576$. In addition, our randomized mechanism is universally strategyproof (i.e., it is a lottery over deterministic strategyproof mechanisms) while the one by~\citeauthor{chen2011mechanism} is only strategyproof in expectation. Universally strategyproof mechanisms are more desirable since they keep their strategyproofness even when the agents are informed about the realization of the internal 
randomness of the mechanism.

Both our deterministic mechanism and our randomized mechanism are based on a strategyproof version of the \textit{greedy} algorithm. Our basic greedy mechanism allocates each agent a quota, corresponding to the total size of their items in the fractional greedy solution,\footnote{That is, the optimal fractional solution constructed by packing items in decreasing order of their density, breaking ties in favor of larger values, until the capacity is reached or there are no more items.} and packs the individual optimal solution of each agent constrained to their quota. In order to achieve our deterministic guarantee, we check whether an agent owns $2/3$ of the value of the fractional greedy solution, pack their optimum if this is the case, and run our greedy mechanism otherwise. As for the randomized mechanism, we run our greedy mechanism with probability $1/2$ and pick the most valuable item overall with the remaining probability.
It is worth noting that the randomized mechanism by~\citeauthor{chen2011mechanism} is based on a fractional optimal solution as well, but to achieve strategyproofness in expectation they sample a strict subset of the agents and consider a smaller capacity than the original one, ultimately resulting in a worse approximation factor.

As our third contribution, we study the setting in which all items have unit density, i.e., equal value as size. This assumption arises naturally in certain settings; e.g., for a funding agency it is typical that many projects obtain excellent reviews so that it may be difficult to assess the exact value of a project. Thus, it is reasonable to assume that monetary units are equally well spent on all projects. Similarly, ethical concerns may prevent a hospital from associating values with a specific MRI scan making it the primary objective to use the MRI equipment at capacity, a goal that is ultimately achieved by a knapsack solution with unit density items.

We show that the upper bounds of \citeauthor{FeigenbaumJ17} on the approximation guarantees that both deterministic and randomized strategyproof mechanisms can achieve carry over to this setting, by constructing respective instances in \Cref{theo:unit-density-det-LB} and~\Cref{theo:unit-density-rand-LB}. However, we are able to exploit the fact that the values and sizes are the same to construct strategyproof mechanisms with improved guarantees compared to the general setting. Specifically, we give a deterministic strategyproof mechanism with the best-possible approximation guarantee of $1/\phi \approx 0.618$; see \Cref{theo:unit-density-det-UB}. Our mechanism builds upon the fact that, when sorting items in decreasing order of their value, our greedy mechanism performs well as long as the two most valuable items fit together. If an item is valuable enough to provide a $\frac{1}{\phi}$-approximation on its own, we simply pack it. Otherwise, we look for the most valuable item that fits together with each of the less valuable items of other agents; call it $i^*$. We then run our greedy mechanism on the set of items consisting of $i^*$ and all the items that individually fit together with $i^*$. The approximation guarantee is proven via a careful analysis of the relation between items considered in this restricted set and other items. As for strategyproofness, we introduce a notion of dominance between sets, from the perspective of a specific agent, that is based on lexicographic comparisons between the values of the items in the restricted sets before and upon the deletion of an item of this agent.

Finally, we improve the approximation guarantee of our mechanism in the unit density setting by incorporating randomization. Specifically, we design a strategyproof and $\frac{2}{3}$-approximate randomized mechanism; see \Cref{theo:unit-density-rand-UB}. It runs a variant of the previous deterministic mechanism with probability $2/3$, and a similar mechanism that always keeps the most valuable item in the restricted set with the remaining probability. The bound is proven via better-than-$\frac{2}{3}$ guarantees for each of these mechanisms on instances on which the other is not $\frac{2}{3}$-approximate.

{\scriptsize
\begin{table}[tb]
\begin{tabular}{@{}l r @{~} l @{\quad} r @{~} l @{\quad} r @{~} l @{\quad} r @{~} l @{}}
\toprule \\[-8pt]
             & \multicolumn{4}{c}{Deterministic}                                               & \multicolumn{4}{c}{Randomized}                                                                                \\[2pt]
             & \multicolumn{2}{c}{Lower bound}                            & \multicolumn{2}{c}{Upper bound}                            & \multicolumn{2}{c}{Lower bound}                             & \multicolumn{2}{c}{Upper bound}                               \\[3pt] \midrule \\[-7pt]
General, $n=2$ & {\color{blue}$1/2 = 0.5\phantom{00}$} & {\color{blue}[C.~\ref{cor:det-LB-1n}]} & $1/\phi \approx 0.618$ & \cite{FeigenbaumJ17} & $7/(5\!+\!4\sqrt{2})\approx 0.656$ & \cite{FeigenbaumJ17} & $1/(5\phi\!-\!7) \approx 0.917$ & \cite{FeigenbaumJ17}\\[5pt]
\multirow{2}{*}{General, $n \in \N$}      & \multirow{2}{*}{\color{blue}$1/3 \approx 0.333$} & \multirow{2}{*}{\color{blue}[T.~\ref{theo:general-det-UB}]}   & \multirow{2}{*}{$1/\phi \approx 0.618$} & \multirow{2}{*}{\cite{FeigenbaumJ17}}                       & {\color{gray}$1/576 \approx 0.002$} & {\color{gray}\cite{chen2011mechanism,choi2021}}  & \multirow{2}{*}{$1/(5\phi\!-\!7) \approx 0.917$} & \multirow{2}{*}{\cite{FeigenbaumJ17}}                              \\
& & & & & {\color{blue}$1/2 = 0.5\phantom{00}$} & {\color{blue}[T.~\ref{theo:general-rand-UB}]} \\[5pt]
Unit density, $n \in \N$ & {\color{blue} $1/\phi \approx 0.618$} & {\color{blue}[T.~\ref{theo:unit-density-det-UB}]} & {\color{blue}$1/\phi \approx 0.618$} & {\color{blue}[T.~\ref{theo:unit-density-det-LB}]} & {\color{blue}$2/3 \approx 0.667$} & {\color{blue}[T.~\ref{theo:unit-density-rand-UB}]}                                & {\color{blue}$1/(5\phi\!-\!7) \approx 0.917$} & {\color{blue}[T.~\ref{theo:unit-density-rand-LB}]}\\[5pt]
\bottomrule\\
\end{tabular}
\caption{
\label{tab:results}
Our and previous bounds on the approximation guarantees for deterministic and randomized strategyproof mechanisms in different settings. Bounds in blue are obtained in this paper and the number in square brackets refers to the theorem (T.) or corollary (C.) where this bound is shown. Bounds in gray or black are previous results and the number in square brackets refers to the paper where this bound is taken from; bounds in gray are superseded by bounds in this paper.}
\end{table}}

\subsection{Further Related Work}

The topic of this paper belongs to the growing literature in the area of \emph{approximate mechanism design without money}, an agenda put forth in seminal work of \citet{ProcacciaT13}.
Generally speaking, this research area is concerned with applications where a system designer is interested in computing approximately optimal solutions, but the input to the optimization problem is controlled by strategic agents who will strategize their input to the mechanism as long as it is beneficial for them. Most crucially, out of financial or ethical considerations, no money can be used in order to align the incentives of the mechanism designer and the agents.
In the light of the Gibbard-Satterthwaite-Theorem~\citep{Gibbard73,Satterthwaite75}, one cannot hope for general positive results, hence, research has focused on specific optimization problems.
In facility location problems~\citep{AlonFPT10,FeldmanW13,FeigenbaumSY17,GoelH23} the system designer decides on the location of one or more facilities which are close to the agent's locations reported to the mechanism. In general assignment problems~\citep{Bochet12,chen2011mechanism,ChenGL14,FadaeiB17,DughmiG10} a finite set of items has to be assigned to a finite set of bins with given capacity and both the value of an item and its size depend on the bin it is assigned to. Every item corresponds to a strategic agent who can misreport the set of bins it is compatible with.
\citet{FeigenbaumJ17} also consider a model where a knapsack is to be filled with items owned by strategic agents and agents may report fake items.
The issue of strategic agents hiding information from the mechanism is further present in models for kidney exchange. Here, the agents correspond to hospitals that may match some of their patient-donor pairs internally instead of reporting them to the database. Strategyproof matching mechanisms incentivize the hospitals to never match internally; their approximation properties have been studied extensively~\citep{AshlagiR14,AshlagiFKP15,ToulisP11,CaragiannisFP15}. Mechanism design has further been studied in the context of progeny maximization~\citep{BabichenkoDT20,ZhangZZ21} where strategyproof mechanisms incentivize agents to reveal links to other agents. Besides mechanisms without money, incentive issues in situations where a subset of items owned by strategic agents is to be selected in the presence of knapsack constraints have also been studied in the context of auctions~\citep{aggarwal2006knapsack,aggarwal2022simple,jarman2017ex,niazadeh2022fast}.

\section{Preliminaries}
\label{sec:preliminaries}

We let $\N$ denote the strictly positive natural numbers and, for $n\in \N$, $[n]=\{1,\ldots,n\}$. An instance of the strategic knapsack problem is given by a number of agents $n\in \N$, a set of items $E$ with $|E|=m$ partitioned into $n$ sets $E_1,\ldots,E_n$, vectors $v$ and $s$ in $\R^m_{>0}$ whose $i$th component corresponds to the value and the size of the $i$th item in $E$, respectively, and a capacity $C>0$.  
We let $\calI$ denote the set of all such instances. We say that a subset of instances $\calI'\subseteq \calI$ is \textit{closed under inclusion} if, for any instance in $\calI'$ with items $E$, taking any subset of $E$ and keeping the other components of the instance constant leads to another instance in $\calI'$. 
For $i\in E$, we write $a(i)$ for the \textit{owner} of $i$, i.e., the agent such that $i\in E_{a(i)}$.
For $a\in E$ and $E'\subseteq E$, we write $E'_a = E' \cap E_a$ for the items in $E'$ that belong to agent $a$.
For a subset of agents $A \subseteq [n]$, we write $E_A = \bigcup_{a \in A} E_a$ for the set of items of all agents in $A$.
The density of an item $i$ is $d_i = v_i/s_i$; we denote the set of all instances where $d_i = 1$ for all $i$ by $\calI_{\mathrm{UD}}$. Note that $\calI_{\mathrm{UD}}$ is closed under inclusion.
For a subset of items $E'\subseteq E$, we let $v(E') = \sum_{i\in E'} v_i$ and $s(E') = \sum_{i\in E'} s_i$ 
denote the total value and size of this subset, respectively. Further, for a subset $E'\subseteq E$ and an agent $a\in [n]$, we let $v_a(E')=v(E'\cap E_a)$ and  $s_a(E')=s(E'\cap E_a)$ denote the total value and size of the items of $a$ in this subset, respectively. 
Throughout the paper, we assume that (i) no item exceeds the capacity, i.e., $s_i \leq C$ for every $i \in E$; (ii) no pair of items have the same value or size, i.e., $v_i \not= v_j$ and $s_i \not= s_j$ for every distinct $i,j \in E$; and (iii) optimal solutions for (reduced) knapsack problems are unique, i.e., for any instance with items $E$ and capacity $C$, $\arg\max\big\{ v(D) \;\big\vert\; s(D)\leq C', D\subseteq E'\big\}$ is unique for every $E'\subseteq E$ and $C'\leq C$.\footnote{Note that these assumptions are without loss of generality. This is trivial for (i); for (ii) and (iii) one can define small perturbations on the values and sizes such that ties are avoided, e.g., reducing both the value and the size of each item with index $j\in [m]$ by $\varepsilon^j$ for $\varepsilon>0$ sufficiently small. Alternatively, one can employ a fixed tie-breaking rule as done, e.g., in \citet{DisserKMS17}.}

A mechanism for the strategic knapsack problem is a (possibly random) function $\scM\colon E \to 2^E$ that takes a set of items $E$ and outputs a feasible subset of it, i.e., for every $E$ it holds that $s(\scM(E)) \leq C$.\footnote{Throughout the paper, we slightly overload notation by writing $E$ as the unique input of a mechanism and consider its partition into agents, as well as the values and sizes, as attached to this set, and the capacity as fixed.} We say that $\scM$ is \textit{deterministic} if, for every set of items $E$, there exists $E' \subseteq E$ with $\P{\scM(E)=E'}=1$. 
We call a deterministic mechanism \textit{strategyproof} on $\calI' \subseteq \calI$ if for every instance in $\calI'$ with items $E$, every agent $a\in [n]$, and every $E'\subseteq E$ with $E \setminus E' \subseteq E_a$, it holds that $v_a(\scM(E)) \geq v_a(\scM(E'))$. A randomized mechanism is (universally) strategyproof if it is a lottery over strategyproof deterministic mechanisms. 
We let
\begin{equation*}
    \OPT(E) = \arg\max\{v(E') \mid E' \subseteq E \text{ s.t. } s(E') \leq C\}
\end{equation*}
denote the optimal solution for an instance.
A mechanism~$\scM$ is \textit{$\alpha$-approximate} on $\calI' \subseteq \calI$, for \mbox{$\alpha \in [0,1]$}, if for every instance in $\calI'$ it holds that \mbox{$\E{v(\scM(E))} \geq \alpha\cdot v(\OPT(E))$}.
When a mechanism is strategyproof or $\alpha$-approximate on $\calI$, we say that it is strategyproof or \mbox{$\alpha$-approximate}, respectively.
The following lemma states that it suffices to consider agents hiding a single item to show that a mechanism is strategyproof. 

\begin{lemma}
    \label{lem:one-item-suffices}
    Let $\scM$ be a deterministic mechanism that is not strategyproof on $\calI' \subseteq \calI$ for some~$\calI'$ closed under inclusion. Then, there exists an instance in $\calI'$ with items $E$, an agent $a \in [n]$, and an item $i \in E_a$ such that $v_a(\scM(E)) < v_a(\scM(E \setminus \{i\}))$.
\end{lemma}

\begin{proof}
    Let $\scM$ be a deterministic mechanism that is \textit{not} strategyproof on $\calI'\subseteq \calI$ for some subset of instances $\calI'$ closed under inclusion. That is, there exists an instance in $\calI'$ with items $E$, an agent $a \in [n]$, and a subset $E' \subseteq E$ with $E\setminus E' \subseteq E_a$ such that $v_a(\scM(E)) < v_a(\scM(E'))$. Let $d=|E \setminus E'|$ be the number of items in $E$ that are not in $E'$ and let $E \setminus E' = \{j_1,\ldots, j_d\}$ denote the set of these items in an arbitrary order. For $q \in \{0,\ldots,d\}$, we also let
    \begin{equation*}
        E^q = E \setminus \bigcup_{p=1}^{q} \{j_p\}
    \end{equation*}
    denote the set obtained by deleting the first $q$ items in $E \setminus E'$ from $E$; in particular, $E^0=E$ and $E^d=E'$. We know that
    \begin{equation*}
        v_a(\scM(E^0)) -  v_a(\scM(E^d)) = \sum_{p=1}^d \big(v_a(\scM(E^{p-1})) -  v_a(\scM(E^p))\big) < 0.
    \end{equation*}
    Therefore, for some $p\in [d]$ it holds that $v_a(\scM(E^{p-1})) -  v_a(\scM(E^p)) < 0$ , which is equivalent to \mbox{$v_a(\scM(E^{p-1})) < v_a(\scM(E^{p-1} \setminus \{j_p\}))$}. Since $\calI'$ is closed under inclusion, we conclude.
\end{proof}

\subsection{Greedy Solution}

When developing and analyzing our mechanisms, the notion of the greedy solution for the knapsack problem plays an important role. 
For a set of items $E$, we denote the items sorted by density in decreasing order, breaking ties in favor of larger values, by $\smash{E=\{i_1(E),i_2(E),\ldots,i_m(E)\}}$, such that $\smash{d_{i_1(E)} \geq d_{i_2(E)} \geq  \dots \geq d_{i_m(E)}}$ and whenever $\smash{d_{i_k(E)}=d_{i_{k+1}(E)}}$ it holds that $\smash{v_{i_k(E)} > v_{i_{k+1}(E)}}$. We refer to this order of items as the \textit{greedy order} of set $E$. We further let $\smash{\ell(E) = \max\bigl\{q \in [m] \mid \sum_{p=1}^{q} s_{i_p(E)} \leq C\bigr\}}$ denote the largest index 
of an item that fully fits in the knapsack, and define $\smash{x(E)\in [0,1]^E}$ as
\begin{equation*}
    x_{i_k(E)}(E) = \begin{cases} 1 & \text{if } k \leq \ell(E),\\ \dfrac{C - \sum_{p=1}^{\ell(E)} s_{i_p(E)}}{s_{i_k(E)}} & \text{if } k=\ell(E)+1, \qquad \text{for every }k \in [m]. \\ 0 & \text{otherwise,}\end{cases}
\end{equation*}
Note that these values are uniquely defined. We refer to this particular (fractional) solution as the \textit{fractional greedy solution} throughout the paper. We further refer to the set \mbox{$\{i\in E \mid x_i(E)=1\}$} as the \textit{integral greedy solution}. When the set of items $E$ is clear from the context, we omit this argument from the previous definitions and simply write $\smash{i_p}$, $\ell$, and $\smash{x_i}$.

It is easy to see that, for every $E$, $x$ is such that $\smash{\sum_{i \in E} s_ix_i = \min \big\{ C,~ \sum_{i \in E} s_i \big\}}$.
It is also well known that, for any instance, $x$ defines an optimal solution for the \textit{fractional} version of the knapsack problem~\citep[e.g.,][\S~17.1]{KorteV18}; i.e., it is an optimal solution of the following linear program with variables $y_i \in [0,1]$ for every $i\in E$:
\begin{equation}
    \max \bigg\{ \sum_{i \in E}v_i y_i\;\bigg\vert\; \sum_{i \in E} s_i y_i \leq C,\;y \in [0,1]^E \bigg\}.\label{eq:lp}
\end{equation}
We state the following properties of the fractional greedy solution for future reference.

\begin{proposition}
    \label{pro:fractional-greedy}
    For a set of items $E$ with sizes $s\in \R^E_{>0}$ and values $v \in \R^E_{>0}$, it holds that
    \begin{enumerate}[label=(\roman*)]
        \item $\sum_{i \in E} s_ix_i = \min \big\{ C,~ \sum_{i \in E} s_i \big\}$, \label{item-fractional-greedy-size}
        \item $x$ is an optimal solution for the linear program \eqref{eq:lp},\label{item-fractional-greedy-value}
        \item $\sum_{i\in E} v_ix_i \geq \OPT(E)$.\label{item:fractional-greedy-opt}
    \end{enumerate}
\end{proposition}

We state here one more property regarding the fractional greedy solution: Whenever we have two instances $E$ and $E'$ and an item $i \in E \cap E'$ such that the total size of items that are in front of $i$ in the greedy order of $E'$ is smaller than the total size of items that are in front of $i$ in the greedy order of $E$, it holds that $x_i(E)\leq x_i(E')$.

\begin{lemma}
    \label{lem:fractional-greedy-improvement}
    Let $E$ and $E'$ be sets of items, $i\in E\cap E'$, and $k,k'$ be such that $i=i_k(E)=i_{k'}(E')$. If $\sum_{p=1}^{k-1} s_{i_p(E)} \geq \sum_{p=1}^{k'-1} s_{i_{p}(E')}$,
    then $x_i(E)\leq x_i(E')$.
\end{lemma}

\begin{proof}
    Let $E,E',i,k$, and $k'$ be as in the statement.
    If $x_i(E)=0$, or if both $x_i(E)<1$ and $x_i(E')=1$ hold, the result follows trivially. If $x_i(E)\in (0,1)$ and $x_i(E')<1$, we can write
    \begin{equation*}
        x_i(E) = \frac{C-\sum_{p=1}^{k-1}s_{i_p(E)}}{s_i} \leq \frac{C-\sum_{p=1}^{k'-1}s_{i_p(E')}x_{i_p(E')}}{s_i} = x_i(E'),
    \end{equation*}
    where we used the hypothesis in the statement and the fact that $x_j(E')\leq 1$ for every $j\in E'$. Finally, if $x_i(E)=1$ we know that $\sum_{p=1}^{k}s_{i_p(E)} \leq C$, hence from the hypothesis we obtain that $\sum_{p=1}^{k'}s_{i_p(E')} \leq C$. Thus, $x_i(E')=1$ as well.
\end{proof}

\section{General Knapsack}\label{sec:general}

In this section, we study the strategic knapsack problem with arbitrary item densities.
We provide the first deterministic mechanism achieving a constant-factor approximation and we improve the approximation factor given by a randomized mechanism from $1/576$ to $1/2$.

These results have as a main building block a strategyproof version of the algorithm that for every instance simply returns the \textit{integral greedy solution}. It is not hard to see that the integral greedy solution is not strategyproof (\Cref{fig:greedy-not-sp} shows a simple example of this fact) nor does it provide any finite approximation guarantee. However, as we will see, a small modification allows us to recover strategyproofness while packing a set of items with at least as much value as the integral greedy solution for every instance.
\begin{figure}[t]
\centering
\begin{tikzpicture}
\draw[fill = lightblue] (0,0) rectangle node[label={[label distance=0.1cm]-90:$1$}] {$10$} (0.5,0.7);
\draw[fill = lightgreen] (0.5,0) rectangle node[label={[label distance=0.1cm]-90:$3$}] {$25$} (2,0.7);
\draw[fill = lightorange] (2,0) rectangle node[label={[label distance=0.1cm]-90:$3$}] {$20$} (3.5,0.7);   
\draw[fill = lightblue] (3.5,0) rectangle node[label={[label distance=0.1cm]-90:$1$}] {$6$} (4,0.7);
\draw[fill = lightorange] (4,0) rectangle node[label={[label distance=0.1cm]-90:$3$}] {$15$} (5.5,0.7);
\draw[fill = lightorange] (5.5,0) rectangle node[label={[label distance=0.1cm]-90:$2$}] {$8$} (6.5,0.7);
\draw[fill = lightgreen] (6.5,0) rectangle  node[label={[label distance=0.1cm]-90:$5$}] {$5$} (9,0.7);

\draw [very thick, -](5,-0.2) -- (5,0.9) node[above] {$C=10$};
\Text[x=-0.3,y=-0.23]{$s$};
\Text[x=-0.3,y=0.33]{$v$};

\begin{scope}[shift={(0,-1.4)}]
\draw[fill = lightblue] (0,0) rectangle node[label={[label distance=0.1cm]-90:$1$}] {$10$} (0.5,0.7);
\draw[fill = lightgreen] (0.5,0) rectangle node[label={[label distance=0.1cm]-90:$3$}] {$25$} (2,0.7);
\draw[fill = lightorange] (2,0) rectangle node[label={[label distance=0.1cm]-90:$3$}] {$20$} (3.5,0.7);   
\draw[fill = lightblue] (3.5,0) rectangle node[label={[label distance=0.1cm]-90:$1$}] {$6$} (4,0.7);
\draw[fill = lightorange] (4,0) rectangle node[label={[label distance=0.1cm]-90:$2$}] {$8$} (5,0.7);
\draw[fill = lightgreen] (5,0) rectangle  node[label={[label distance=0.1cm]-90:$5$}] {$5$} (7.5,0.7);

\draw [very thick, -](5,-0.2) -- (5,0.9);
\Text[x=-0.3,y=-0.23]{$s$};
\Text[x=-0.3,y=0.33]{$v$};

\end{scope}
\end{tikzpicture}
\caption{Example of an instance in which computing the integral greedy solution results in an agent benefiting from hiding items. Items of the same color belong to the same agent. The value and size of each item are respectively written inside and below them; the capacity is $10$. If the agent that owns the orange items hides the item with value $15$ in the upper instance, the value of the orange items packed in the integral greedy solution increases from $20$ to $28$.}
\label{fig:greedy-not-sp}
\end{figure}
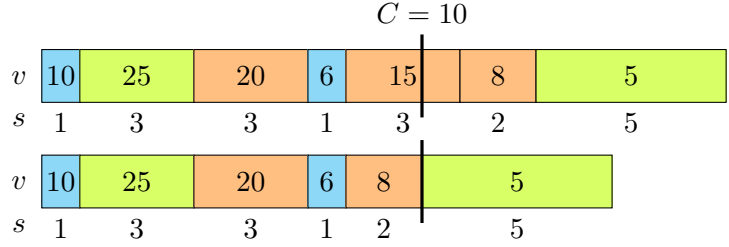

We call the resulting mechanism \nalgspgreedy, abbreviate it \algspgreedy,~and formally describe it in \Cref{alg:spgreedy}.
The mechanism first assigns to each agent $a \in [n]$ a quota equal to $\sum_{i \in E_a} s_ix_i$, i.e., the total size of the agent's items in the fractional greedy solution. 
It then solves an integer knapsack problem for each agent, where the capacity of the agent is equal to their quota, and finally returns all items contained in these integer solutions.
\begin{algorithm}[t]
\caption{\nalgspgreedy~$(\algspgreedy)$}\label{alg:spgreedy}
\begin{algorithmic}
\Require set of items $E$ (partitioned into $E_1,\ldots,E_n$ with values $v$ and sizes $s$), capacity $C$
\Ensure subset $E^* \subseteq E$ with $s(E^*) \leq C$
\State $x \gets \text{fractional greedy solution}$
\ForEach{$a \in [n]$}
\State $E^*_a \gets \arg\max\big\{ v(E')~\big|~E'\subseteq E_a \text{ s.t. } s(E') \leq \sum_{i\in E_a} s_i x_i\big\}$ 
\EndFor
\State \textbf{return} $\bigcup_{a\in [n]} E^*_a$ 
\end{algorithmic}
\end{algorithm}

The following lemma establishes the strategyproofness of \nalgspgreedy~and shows that the total value of the items selected by the mechanism is at least as high as the value of the integral greedy solution.
The latter holds because an agent can always pack their items that belong to the integral greedy solution; by choosing their optimum constrained to a capacity $\sum_{i \in E_a} s_ix_i$, they can only increase the value of their packed items.

\begin{lemma}
    \label{lem:spgreedy}
    \nalgspgreedy~is strategyproof. For every set of items $E$, $v(\algspgreedy(E)) \geq \sum_{i \in E : x_i(E)=1} v_i$.
\end{lemma}

\begin{proof}
    We first show strategyproofness. Let $E$ be an arbitrary set of items, $a \in [n]$, and $i \in E_a$. We distinguish three cases based on the relation between the sizes of the items and the capacity of the knapsack.
    
    If $\sum_{j \in E}s_j < C$, then
    \begin{equation*}
        \sum_{j\in E_a} s_jx_j(E) = \sum_{j\in E_a} s_j \geq \sum_{j\in E_a \setminus \{i\}} s_j = \sum_{j\in E_a \setminus \{i\}} s_jx_j(E \setminus \{i\}),
    \end{equation*}
    where the first and last equalities come from property \ref{item-fractional-greedy-size} of \Cref{pro:fractional-greedy}. 
    
    If $\sum_{j \in E \setminus \{i\}}s_j < C \leq \sum_{j \in E}s_j$, then
    \begin{align*}
        \sum_{j\in E_a} s_jx_j(E) & = C - \sum_{j\in E \setminus E_a} s_jx_j(E) &  \sum_{j \in E}s_j\geq C \text{ and property \ref{item-fractional-greedy-size} of \Cref{pro:fractional-greedy}} \\
        & \geq C - \sum_{j\in E \setminus E_a} s_j & x_j(E) \leq 1 \text{ for every } j \in E\\
        & > \sum_{j \in E \setminus \{i\}} s_j - \sum_{j\in E \setminus E_a} s_j & \sum_{j \in E \setminus \{i\}}s_j < C\\
        & = \sum_{j\in E_a \setminus \{i\}} s_jx_j(E\setminus \{i\}) & \sum_{j \in E \setminus \{i\}} s_j < C \text{ and property \ref{item-fractional-greedy-size} of \Cref{pro:fractional-greedy}.}
    \end{align*}
    
    We finally consider the case with $\sum_{j \in E \setminus \{i\}}s_j \geq C$. For any fixed $j\in E\setminus E_a$, defining $k,k'$ such that $j=i_k(E)=i_{k'}(E\setminus \{i\})$, we have that $\{i_p(E\setminus \{i\}) \mid p< k'\} \subseteq \{i_p(E) \mid p< k\}$.
    Thus, \Cref{lem:fractional-greedy-improvement} implies that for every $j \in E \setminus E_a$, we have that $x_j(E \setminus \{i\}) \geq x_j(E)$. This yields
    \begin{equation*}
        \sum_{j\in E_a} s_jx_j(E) = C - \sum_{j\in E \setminus E_a} s_jx_j(E) \geq C - \sum_{j\in E \setminus E_a} s_jx_j(E \setminus \{i\}) = \sum_{j\in E_a \setminus \{i\}} s_jx_j(E\setminus \{i\}),
    \end{equation*}
    where we used the fact stated above and property \ref{item-fractional-greedy-size} of \Cref{pro:fractional-greedy}. 
    
    We conclude that, no matter the relation between the sizes of the items and the capacity, the inequality $\sum_{j\in E_a} s_jx_j(E) \geq \sum_{j \in E_a \setminus \{i\}} s_jx_j(E\setminus \{i\})$ holds. Therefore,
    \begin{align*}
        v_a(\algspgreedy(E)) & = \max\bigg\{v(E')~\bigg|~E' \subseteq E_a \text{ s.t. } s(E') \leq \sum_{j\in E_a} s_jx_j(E) \bigg\} \\
        & \geq \max\bigg\{v(E')~\bigg|~E' \subseteq E_a \setminus \{i\} \text{ s.t. } s(E') \leq \sum_{j\in E_a \setminus \{i\}} s_jx_j(E \setminus \{i\}) \bigg\} \\
        & = v_a(\algspgreedy(E \setminus \{i\})).
    \end{align*}

    We now prove the second, simpler part of the statement. For every $a \in [n]$,
    \begin{equation*}
        v_a(\algspgreedy(E)) = \max\bigg\{v(E')~\bigg|~E' \subseteq E_a \text{ s.t. } s(E') \leq \sum_{i\in E_a} s_ix_i(E) \bigg\} \geq v( \{ i\in E_a \mid x_i(E) = 1\} ).
    \end{equation*}
    The result follows directly by summing over all agents.
\end{proof}

\nalgspgreedy\ is not $\alpha$-approximate for any $\alpha > 0$. 
In the following, we use \nalgspgreedy~as a building block of deterministic and randomized mechanisms with better approximation guarantees.

\subsection{Deterministic Mechanisms}\label{subsec:general-det}

For deterministic strategyproof mechanisms, no lower bound on the approximation guarantee they can provide was known. As an initial observation, a $\smash{\frac{1}{n}}$-approximation is easy to get through a mechanism that returns the best individual optimum of an agent. A natural question, however, regards the existence of a strategyproof mechanism providing a constant approximation factor for an arbitrary number of agents. The following result answers this question affirmatively.

\begin{theorem}
    \label{theo:general-det-UB}
    There exists a strategyproof and $\smash{\frac{1}{3}}$-approximate deterministic mechanism.
\end{theorem}

Our mechanism is based on \nalgspgreedy. 
When $E$ is such that the integral greedy solution provides a value of $\alpha\cdot v(\OPT(E))$, \nalgspgreedy~is $\alpha$-approximate due to \Cref{lem:spgreedy}. 
However, this mechanism can be arbitrarily bad when most of the value of the fractional greedy solution comes from the \textit{fractional item} $i\in E$ with $0<x_i<1$. 
To address this case, if an agent owns a fraction of at least $2/3$ of the value of the fractional greedy solution, we return their optimal solution. If there is no such agent, we return $\algspgreedy(E)$.
The resulting mechanism, \nalgmixspgreedy~($\algmixspgreedy$), is presented as \Cref{alg:mixspgreedy}.

\begin{algorithm}[t]
\caption{\nalgmixspgreedy~$(\algmixspgreedy)$}\label{alg:mixspgreedy}
\begin{algorithmic}
\Require set of items $E$ (partitioned into $E_1,\ldots,E_n$ with values $v$ and sizes $s$), capacity $C$
\Ensure subset $E^* \subseteq E$ with $s(E^*) \leq C$
\State $x \gets \text{fractional greedy solution}$
\If{there exists $a\in [n]$ such that $\sum_{i \in E_a}v_ix_i \geq \frac{2}{3} \sum_{i \in E}v_ix_i$}
    \State \textbf{return} $\OPT(E_a)$
\Else
    \State \textbf{return} $\algspgreedy(E)$
\EndIf
\end{algorithmic}
\end{algorithm}

We now prove the approximation guarantee and strategyproofness of \nalgmixspgreedy. The former follows by carefully analyzing both possible outputs of the mechanism. If an agent $a$ is such that $\smash{\sum_{i \in E_a}v_ix_i(E) \geq \frac{2}{3} \sum_{i \in E}v_ix_i(E)}$, then either the items $i\in E_a$ with $x_i=1$ or an item $i \in E_a$ with $0<x_i<1$ achieve half of this value, and they both define feasible solutions. Otherwise, the integral greedy solution provides a $\smash{\frac{1}{3}}$-approximation and this bound carries over to \nalgspgreedy~due to \Cref{lem:spgreedy}.
To show the strategyproofness of \nalgmixspgreedy, we first prove that, whenever an agent hides an item, the fraction of the value of this agent's items in the fractional greedy solution cannot increase and the fraction of the value of other agents' items in the fractional greedy solution cannot decrease. 

\newpage
\begin{lemma}
    \label{lem:fractional-greedy-hiding}
    For every set of items $E$, agent $a \in [n]$, and item $i \in E_a$, the following inequalities hold:
    \begin{enumerate}[label=(\roman*)]
        \item $\dfrac{\sum_{j \in E_b}v_jx_j(E)}{\sum_{j \in E}v_jx_j(E)} \leq \dfrac{\sum_{j \in E_b}v_jx_j(E \setminus \{i\})}{\sum_{j \in E \setminus \{i\}}v_jx_j(E \setminus \{i\})}$ for every $b \in [n] \setminus \{a\}$,\label{item:fractional-greedy-hiding-others}
        \item $\dfrac{\sum_{j \in E_a}v_jx_j(E)}{\sum_{j \in E}v_jx_j(E)} \geq \dfrac{\sum_{j \in E_a \setminus \{i\}}v_jx_j(E \setminus \{i\})}{\sum_{j \in E \setminus \{i\}}v_jx_j(E \setminus \{i\})}$. \label{item:fractional-greedy-hiding-liar}
    \end{enumerate}
\end{lemma}

\begin{proof}
    Let $E,~a$, and $i$ be as in the statement. To see part \ref{item:fractional-greedy-hiding-others}, let $b\in [n]\setminus \{a\}$ and note that, for any fixed $j\in E_b$, defining $k,k'$ such that $j=i_k(E)=i_{k'}(E\setminus \{i\})$, we have that $\{i_p(E\setminus \{i\}) \mid p< k'\} \subseteq \{i_p(E) \mid p< k\}$. Thus, \Cref{lem:fractional-greedy-improvement} implies that for every $j \in E_b$ we have that $x_j(E \setminus \{i\}) \geq x_j(E)$, so we get $\sum_{j \in E_b}v_jx_j(E \setminus \{i\}) \geq \sum_{j \in E_b}v_jx_j(E)$.
    Additionally, we observe that any feasible solution of the linear program \eqref{eq:lp} with a set of items $E\setminus \{i\}$ is feasible for this program with a set of items $E$ as well. Therefore, its optimal value with a set of items $E$ is at least its optimal value with a set of items $E\setminus \{i\}$. Part \ref{item-fractional-greedy-value} of \Cref{pro:fractional-greedy} then yields $\sum_{j \in E}v_jx_j(E) \geq \sum_{j \in E \setminus \{i\}}v_jx_j(E \setminus \{i\})$. Both inequalities together imply part \ref{item:fractional-greedy-hiding-others} of \Cref{lem:fractional-greedy-hiding}.
    
    We directly obtain part \ref{item:fractional-greedy-hiding-liar} by observing that 
    \begin{align*}
        \frac{\sum_{j \in E_a}v_jx_j(E)}{\sum_{j \in E}v_jx_j(E)} &= \frac{\sum_{j \in E} v_jx_j(E) - \sum_{b \in [n] \setminus \{a\}} \sum_{j \in E_b} v_jx_j(E)}{\sum_{j \in E}v_jx_j(E)} \\
        &= 1 - \sum_{b \in [n] \setminus \{a\}} \frac{\sum_{j \in E_b} v_jx_j(E)}{\sum_{j \in E}v_jx_j(E)} \\
        &\geq 1 - \sum_{b \in [n] \setminus \{a\}} \frac{\sum_{j \in E_b}v_jx_j(E \setminus \{i\})}{\sum_{j \in E \setminus \{i\}}v_jx_j(E \setminus \{i\})} \\
        &= \frac{\sum_{j \in E \setminus \{i\}}v_jx_j(E \setminus \{i\}) - \sum_{b \in [n] \setminus \{a\}} \sum_{j \in E_b}v_jx_j(E \setminus \{i\})}{\sum_{j \in E \setminus \{i\}}v_jx_j(E \setminus \{i\})} \\
        &= \frac{\sum_{j \in E_a \setminus \{i\}}v_jx_j(E \setminus \{i\})}{\sum_{j \in E \setminus \{i\}}v_jx_j(E \setminus \{i\})},
    \end{align*}
    where we used part \ref{item:fractional-greedy-hiding-others} of \Cref{lem:fractional-greedy-hiding} for the inequality.
\end{proof}

The previous lemma implies that, if some agent owns at least $2/3$ of the value of the fractional greedy solution, other agents cannot change this fact, while the agent themselves has no incentive to change it. If no agent owns at least $2/3$ of the value of the fractional greedy solution, an agent cannot increase their own value to get over this threshold and \nalgmixspgreedy\ runs \nalgspgreedy, which is strategyproof due to \Cref{lem:spgreedy}.
We are now ready to prove that \nalgmixspgreedy\ is strategyproof and $\smash{\frac{1}{3}}$-approximate.

\begin{proof}[Proof of \Cref{theo:general-det-UB}]
    We claim the result for the \nalgmixspgreedy\ mechanism $\algmixspgreedy$. We first show strategyproofness. Let $E$ be a set of items, $a \in [n]$, and $i \in E_a$. In the following, we distinguish some cases that determine the outcome of \nalgspgreedy~for inputs $E$ and $E\setminus \{i\}$ and show that, for all of them, it holds that $v_a(\algmixspgreedy(E)) \geq v_a(\algmixspgreedy(E \setminus \{i\}))$.

    If there exists $b\in [n] \setminus \{a\}$ such that $\sum_{j \in E_b}v_jx_j(E \setminus \{i\}) \geq \frac{2}{3} \sum_{j \in E \setminus \{i\}}v_jx_j(E \setminus \{i\})$, then $\algmixspgreedy(E\setminus \{i\}) = \OPT(E_b)$ and, therefore,
    \begin{equation} 
        v_a(\algmixspgreedy(E)) \geq 0 = v_a(\algmixspgreedy(E \setminus \{i\})).\label{eq:sp-spgreedy-i}
    \end{equation}
    If there exists $b\in [n] \setminus \{a\}$ such that $\sum_{j \in E_b}v_jx_j(E) \geq \frac{2}{3} \sum_{j \in E}v_jx_j(E)$, part \ref{item:fractional-greedy-hiding-others} of \Cref{lem:fractional-greedy-hiding} implies that
    \begin{equation*}
        \frac{\sum_{j \in E_b}v_jx_j(E \setminus \{i\})}{\sum_{j \in E \setminus \{i\}}v_jx_j(E \setminus \{i\})} \geq \frac{\sum_{j \in E_b}v_jx_j(E)}{\sum_{j \in E}v_jx_j(E)} \geq \frac{2}{3},
    \end{equation*}
    so the chain of inequalities \eqref{eq:sp-spgreedy-i} still holds.
    
    If $\sum_{j \in E_a}v_jx_j(E) \geq \frac{2}{3} \sum_{j \in E}v_jx_j(E)$, then $\algmixspgreedy(E) = \OPT(E_a)$ and we get
    \begin{equation}
        v_a(\algmixspgreedy(E)) = v(\OPT(E_a)) \geq v(\OPT(E_a \setminus \{i\})) \geq v_a(\algmixspgreedy(E \setminus \{i\})).\label{eq:sp-spgreedy-ii}
    \end{equation}
    If $\sum_{j \in E_a \setminus \{i\}}v_jx_j(E \setminus \{i\}) \geq \frac{2}{3} \sum_{j \in E \setminus \{i\}}v_jx_j(E \setminus \{i\})$, we observe by part \ref{item:fractional-greedy-hiding-liar} of \Cref{lem:fractional-greedy-hiding} that
    \begin{align*}
        \frac{\sum_{j \in E_a}v_jx_j(E)}{\sum_{j \in E}v_jx_j(E)} \geq \frac{\sum_{j \in E_a \setminus \{i\}}v_jx_j(E \setminus \{i\})}{\sum_{j \in E \setminus \{i\}}v_jx_j(E \setminus \{i\})} \geq \frac{2}{3}.
    \end{align*}
    Hence, the chain of inequalities \eqref{eq:sp-spgreedy-ii} still holds.
    
    If $\sum_{j \in E_a}v_jx_j(E) < \frac{2}{3} \sum_{j \in E}v_jx_j(E)$ and \mbox{$\sum_{j \in E_b}v_jx_j(E \setminus \{i\}) < \frac{2}{3} \sum_{j \in E \setminus \{i\}}v_jx_j(E \setminus \{i\})$} for every $b\in [n] \setminus \{a\}$,
    then $\algmixspgreedy(E) = \algspgreedy(E)$ and $\algmixspgreedy(E \setminus \{i\}) = \algspgreedy(E \setminus \{i\})$. Therefore,
    \begin{equation*} 
        v_a(\algmixspgreedy(E)) = v_a(\algspgreedy(E)) \geq v_a(\algspgreedy(E \setminus \{i\})) = v_a(\algmixspgreedy(E \setminus \{i\})),
    \end{equation*}
    where the inequality follows from the strategyproofness of \nalgspgreedy~established in \Cref{lem:spgreedy}. We conclude that \nalgmixspgreedy~is indeed strategyproof.
    
    We now prove the approximation guarantee provided by \nalgmixspgreedy. Let $E$ be an arbitrary set of items. We first suppose that there exists $a \in [n]$ with \mbox{$\sum_{i \in E_a}v_ix_i \geq \frac{2}{3} \sum_{i \in E}v_ix_i$}. Recall that $s_i \leq C$ for every $i \in E$ and that we know from part \ref{item-fractional-greedy-size} of \Cref{pro:fractional-greedy} that $\sum_{i \in E_a: x_i=1} s_i \leq C$. Therefore,
    \begin{equation}
        s(\{i \in E_a \mid x_i=1\}) \leq C \text{ and } s(\{i \in E_a \mid 0<x_i<1\}) \leq C,\label{eq:feasible-sets}
    \end{equation}
    where the set $\{i \in E_a: 0<x_i<1\}$ is either empty or a single item. This yields
    \begin{align*}
        v(\algmixspgreedy(E)) & = v(\OPT(E_a)) = \max\{ v(E') \mid E' \subseteq E_a \text{ s.t. } s(E') \leq C\} \\
        & \geq \max\big\{ v(\{i \in E_a \mid x_i=1\}),~v(\{i \in E_a \mid 0<x_i<1\}) \big\} \\
        & \geq \frac{1}{2}\big[v(\{i \in E_a \mid x_i=1\}) + v(\{i \in E_a \mid 0<x_i<1\})\big] \\
        & \geq \frac{1}{2}\sum_{i \in E_a}v_ix_i \geq \frac{1}{2} \cdot \frac{2}{3} \sum_{i \in E}v_ix_i \geq \frac{1}{3}v(\OPT(E)),
    \end{align*}
    where the first inequality follows from expression \eqref{eq:feasible-sets} and the last one from part \ref{item:fractional-greedy-opt} of \Cref{pro:fractional-greedy}. 
    Finally, if $\sum_{i \in E_a}v_ix_i < \frac{2}{3} \sum_{i \in E}v_ix_i$ for every $a\in [n]$, we know in particular that $v_{i_\ell}x_{i_\ell} < \frac{2}{3} \sum_{i \in E}v_ix_i$, thus
    \begin{align*}
        v(\algmixspgreedy(E)) = v(\algspgreedy(E)) \geq \sum_{i \in E: x_i=1} v_i & = \sum_{i \in E}v_ix_i - v_{i_\ell}x_{i_\ell} > \frac{1}{3} \sum_{i \in E}v_ix_i \geq \frac{1}{3}v(\OPT(E)).
    \end{align*}
    Indeed, the first inequality comes from \Cref{lem:spgreedy}, the second inequality from the fact that $\smash{v_{i_\ell}x_{i_\ell} < \frac{2}{3} \sum_{i \in E}v_ix_i}$, and the third inequality from part \ref{item:fractional-greedy-opt} of \Cref{pro:fractional-greedy}. We conclude that \nalgmixspgreedy~is $\frac{1}{3}$-approximate, as claimed.
\end{proof}

We finish by stating the currently best-known bounds for every $n$. The strategyproof and $\frac{1}{n}$-approximate mechanism mentioned at the beginning of this section, that selects the best individual optimum, improves on our $\frac{1}{3}$-approximation for the case of two agents.

\begin{corollary}
\label{cor:det-LB-1n}
    For $n\in \N$, there exists a strategyproof and $\max\{1/3, 1/n\}$-approximate deterministic mechanism on instances with $n$ agents.
\end{corollary}

\subsection{Randomized Mechanisms}\label{subsec:general-rand}

When randomization is allowed, the existence of a strategyproof mechanism providing a constant approximation guarantee was established by \citet{chen2011mechanism}, with the best possible approximation factor known to lie between \mbox{$1/576 \approx 0.0017$}~\citep{chen2011mechanism,choi2021} and $1/(5\phi-7) \approx 0.9173$~\citep{FeigenbaumJ17}. The following result narrows this gap significantly.

\begin{theorem}
    \label{theo:general-rand-UB}
    There exists a strategyproof and $\smash{\frac{1}{2}}$-approximate randomized mechanism.
\end{theorem}

In order to achieve this approximation, we rely once again on \nalgspgreedy. The idea is now simpler: As this mechanism guarantees a value equal to the value of the integral greedy solution, and the optimal solution is at most the value of the integral greedy solution plus the value of the fractional item, we return either the set of items given by \nalgspgreedy~or the most valuable item overall, each with probability $1/2$. We call the resulting mechanism \nalgrandspgreedy~($\algrandspgreedy$) and describe it as \Cref{alg:randspgreedy}. 
\begin{algorithm}[t]
\caption{\nalgrandspgreedy~$(\algrandspgreedy)$}\label{alg:randspgreedy}
\begin{algorithmic}
\Require set of items $E$ (partitioned into $E_1,\ldots,E_n$ with values $v$ and sizes $s$), capacity $C$
\Ensure subset $E^* \subseteq E$ with $s(E^*) \leq C$
\State $X \gets \text{Bernoulli}(1/2)$
\If{$X = 1$}
    \State \textbf{return } $\algspgreedy(E)$
\Else
    \State \textbf{return } $\arg\max \{v_i \mid i \in E\}$
\EndIf
\end{algorithmic}
\end{algorithm}

The strategyproofness of \nalgrandspgreedy~follows directly from the fact that both \nalgspgreedy~and the mechanism that always returns the most valuable item are strategyproof. The approximation guarantee is proven by noting that at least half of the value of the fractional greedy solution is returned in expectation. We now present the formal proof.

\begin{proof}[Proof of \Cref{theo:general-rand-UB}]
    We claim the result for \nalgrandspgreedy. Strategyproofness of \nalgspgreedy~was already established in \Cref{lem:spgreedy}. Moreover, for a set of items $E$, an agent $a\in [n]$, and $i \in E_a$, we have that \mbox{$v_a(\arg\max \{v_j \mid j \in E\}) \geq v_a(\arg\max \{v_j \mid j \in E \setminus \{i\} \})$}. Hence, \nalgrandspgreedy~is a lottery over two strategyproof mechanisms.

    To prove the approximation guarantee, we consider an arbitrary set of items $E$ and observe that
    \begin{align*}
        \E{v(\algrandspgreedy(E))} & \geq \frac{1}{2} ( v(\algspgreedy(E)) + v(\arg\max \{v_j \mid j \in E\}) ) \geq \frac{1}{2} \bigg( \sum_{i \in E: x_i=1} v_i + v_{i_\ell} \bigg) \\
        &\geq \frac{1}{2} \bigg( \sum_{i \in E: x_i=1} v_i + v_{i_\ell}x_{i_\ell} \bigg) = \frac{1}{2} \sum_{i \in E} v_ix_i \geq \frac{1}{2} \OPT(E),
    \end{align*}
    where the second inequality follows from \Cref{lem:spgreedy}, the third one from the fact that $x_{i_\ell} \leq 1$, and the last one from part \ref{item:fractional-greedy-opt} of \Cref{pro:fractional-greedy}.
\end{proof}

\section{Unit-Density Knapsack}\label{sec:unit-density}

In this section, we address the unit-density setting, in which each item's value equals its size. We improve the performance of deterministic mechanisms, providing a tight approximation factor equal to the inverse of the golden ratio, i.e., $\smash{1/\phi = 2/(1+\sqrt{5}) \approx 0.618}$. We further give a randomized mechanism that builds upon the deterministic one and guarantees a \mbox{$\smash{\frac{2}{3}}$-approximation}, thus showing a strict separation between the best-possible approximation factors that deterministic and randomized mechanisms can achieve in this setting.

Both results are based on a family of mechanisms we call \nalgud, parameterized on a value $\beta \in [1/2, 2/3]$. For a specific value of $\beta$, we denote its output for a set of items $E$ as $\algud_\beta(E)$. In simple terms, the mechanism first checks whether a single agent can provide an approximation factor of $\beta$ with their individual optimum and returns the best individual optimum in such a case. Otherwise, it defines $i^*$ as the most valuable item that can be packed with any item that has less value and is owned by another agent. It then considers the set $R$, containing $\{i^*\}$ and all items that, on their own, fit together with $i^*$; we often refer to this set as the \textit{restricted} set of items. Finally, it runs a variant of the \nalgspgreedy~mechanism from \Cref{subsec:general-det}, described as \Cref{alg:rspgreedy}, that sorts the items by value and only considers the restricted set of items when assigning knapsack capacity. We call this auxiliary mechanism \nalgrspgreedy~and abbreviate it as $\algrspgreedy$. Note that in the description of the mechanisms, as well as throughout this section, we make use of the fact that $v_i=s_i$ for any item $i$ and only refer to values.
\begin{algorithm}[t]
\caption{\nalgud~with parameter $\beta\in [1/2, 2/3]$ $(\algud_\beta)$}\label{alg:ud}
\begin{algorithmic}
\Require set of items $E$ (partitioned into $E_1,\ldots,E_n$ with values $v$ and sizes $s$), capacity $C$
\Ensure subset $E^* \subseteq E$ with $s(E^*) \leq C$
\If{$\max_{a \in [n]} v(\OPT(E_a)) \geq \beta\cdot C$}
    \State $a' \gets \arg\max_{a \in [n]} v(\OPT(E_a))$
    \State \textbf{return} $\OPT(E_{a'})$
\EndIf
\State $P \gets \{i \in E \mid v_i+v_j \leq C \text{ for every } j \in E \setminus E_{a(i)} \text{ with } v_i > v_j \}$
\State $i^* \gets \arg\max \{v_i \mid i \in P \}$
\State $R \gets \{ i^* \} \cup \{i \in E \mid v_{i^*} + v_i \leq C\}$
\State \textbf{return} $\algrspgreedy(E, R)$
\end{algorithmic}
\end{algorithm}
\begin{algorithm}[t]
\caption{\nalgrspgreedy~$(\algrspgreedy)$}\label{alg:rspgreedy}
\begin{algorithmic}
\Require set of items $E$ (partitioned into $E_1,\ldots,E_n$ with values $v$ and sizes $s$), subset $Q \subseteq E$, capacity~$C$
\Ensure subset $E^* \subseteq E$ with $s(E^*) \leq C$
\State $x(Q) \gets \text{fractional greedy solution for items $Q$}$
\State for every $a \in [n]$, $E^*_a \gets \arg\max\big\{ v(E')~\big|~E'\subseteq E_a \text{ s.t. } v(E') \leq \sum_{i\in Q_a} v_i x_i(Q)\big\}$
\State \textbf{return} $\bigcup_{a\in [n]} E^*_a$
\end{algorithmic}
\end{algorithm}

We need some notation to analyze the mechanism. Let $P(E)$ and $R(E)$ denote the sets $P$ and $R$ defined in \nalgud~when run with input $E$. Let $i^*(E)$ denote the item $i^*$ when running this mechanism with input $E$ and let $a^*(E)$ denote its owner. Although these values do not only depend on $E$ but also on $\beta$, we will omit this dependence and make $\beta$ clear from the context. As usual, we will omit the argument $E$ when clear from the context.

The following lemma states the main properties of this family of mechanisms, providing the main ingredient to show our positive results for both deterministic and randomized mechanisms.

\begin{lemma}\label{lem:alg-ud-param}
    For $\beta \in [1/2, 2/3]$, \nalgud~with parameter $\beta$ is strategyproof and $\smash{\min \big\{ \beta, \frac{1-\beta}{\beta} \big\}}$-approximate on $\calI_{\mathrm{UD}}$. Moreover, if $v(\algud_\beta(E)) < \beta\cdot v(\OPT(E))$ holds for some set of items $E$ of an instance in $\calI_{\mathrm{UD}}$, then $v(R(E))<\beta\cdot C$ and $v_{i^*(E)} < \beta\cdot \max\{v_i \mid i\in E\}$.
\end{lemma}

Most of this section is devoted to proving the previous lemma; we then exploit to develop a deterministic mechanism in \Cref{subsec:unit-density-det} and a randomized mechanism in \Cref{subsec:unit-density-rand}.
In what follows, we separately address the strategyproofness and the approximation guarantee provided by \nalgud. 

Establishing the strategyproofness of \nalgud\ constitutes the most demanding task and is shown in several steps. The first condition in \Cref{alg:ud} that checks whether an agent provides a $\beta$-approximation on their own cannot lead to incentives to hide any item, as an agent satisfying this condition can, upon deletion of one of their items, either still satisfy the condition but have a weakly worse optimum, or stop satisfying the condition (and no other solution will provide more value to this agent than their own optimum). If this condition is not fulfilled, we show that no agent $a$ can benefit from hiding an item $i\in E_a$; \Cref{lem:one-item-suffices} allows us to conclude. In order to do this, we introduce a notion of $a$-\textit{dominance} between sets of items that implies that the total value of the items owned by agent $a$ in the fractional greedy solution is higher in the dominating set than in the dominated set. We then show that this notion captures the relation between the restricted sets of items before and after an item of $E_a$ is hidden.

More specifically, for sets of items $E$ and $E'$ and $a \in [n]$, we say that \textit{$E$ $a$-dominates $E'$} if
\begin{enumerate}[label=(\roman*)]
    \item $|E_a|\geq |E'_a|$ and $|E\setminus E_a| \leq |E'\setminus E'_a|$;\label{item:a-dominance-i}
    \item for every $k\in [|E'_a|]$, it holds that $v_{i_k(E_a)} \geq v_{i_k(E'_a)}$;\label{item:a-dominance-ii}
    \item for every $k\in [|E\setminus E_a|]$, it holds that $v_{i_k(E\setminus E_a)} \leq v_{i_k(E'\setminus E'_a)}$.\label{item:a-dominance-iii}
\end{enumerate}
In simple words, $E$ $a$-dominates $E'$ if the values of the items in $E_a$ are lexicographically larger than the values of the items in $E'_a$, and the values of all items of agents other than $a$ in $E$ are lexicographically smaller than the values of all items of agents other than $a$ in $E'$.\footnote{We informally refer to lexicographic comparisons in the intuitive way, comparing the values in decreasing order. For instance, that the values of the items in $E_a$ are lexicographically larger than the values of the items in $E'_a$ means that the $k$th most valuable item in $E_a$ is at least as valuable as the $k$th most valuable item in $E'_a$, for every $k \in [ |E'_a| ]$.}
Crucially, this definition of $a$-dominance of a set of items $E$ over another set $E'$ implies that agent $a$ cannot be better off when running \nalgrspgreedy~on $E'$ than on $E$, because the value of the items owned by this agent in the fractional greedy solution for $E'$ is at most the value of the items owned by this agent in the fractional greedy solution for $E$.

\begin{lemma}
    \label{lem:dominance-implication}
    Let $E$ and $E'$ be sets of items of instances in $\calI_{\mathrm{UD}}$ and $a \in [n]$ be an agent such that $E$ $a$-dominates $E'$. Then, $\sum_{i \in E_a} v_i x_i(E) \geq \sum_{i \in E'_a} v_i x_i(E')$.
\end{lemma}

To prove this lemma, we first show a weaker result: If the value of the items of any subset of agents~$A$ lexicographically increases from $E$ to a set $D$ and the items of all other agents remain the same, then the total value of the items owned by agents in $A$ in the fractional greedy solution cannot be smaller in $D$ than in $E$.

\begin{lemma}
    \label{lem:larger-sizes-implication}
    Let $E$ and $D$ be sets of items of instances in $\calI_{\mathrm{UD}}$ and let $A \subseteq [n]$ be a subset of agents such that $|E_A|\leq |D_A|$, $E\setminus E_A = D \setminus D_A$, and for every $k\in [ |E_A| ]$ it holds that $v_{i_k(E_A)} \leq v_{i_k(D_A)}$. Then 
    \begin{equation*}
        \sum_{i \in E_A} v_ix_i(E) \leq \sum_{i \in D_A} v_ix_i(D).
    \end{equation*}
\end{lemma}

\begin{proof}
Let $E,~ D$, and $A$ be as in the statement. Observe that, for every $k \in [ |E \setminus E_A |]$, it holds that
\begin{align*}
    \sum_{j \in E: v_j > v_{i_k(E \setminus E_A)}} v_j & = \sum_{j \in E_A: v_j > v_{i_k(E \setminus E_A)}} v_j + \sum_{j \in E \setminus E_A: v_j > v_{i_k(E \setminus E_A)}} v_j \\
    & \leq \sum_{j \in D_A: v_j > v_{i_k(D \setminus D_A)}} v_j + \sum_{j \in D \setminus D_A: v_j > v_{i_k(D \setminus D_A)}} v_j \\
    & = \sum_{j \in D: v_j > v_{i_k(D \setminus D_A)}} v_j,
\end{align*}
where the inequality follows from the hypotheses in the statement. Applying \Cref{lem:fractional-greedy-improvement}, this yields $x_{i_k(E \setminus E_A)}(E) \geq x_{i_k(D \setminus D_A)}(D)$ for every $k \in [ |E \setminus E_A |]$ or, in simpler terms, $x_i(E) \geq x_i(D)$ for every $i \in E \setminus E_A$. 

We shall now distinguish three cases. If $\sum_{i\in D}v_i < C$, then we know from the statement that $\sum_{i\in E}v_i < C$ holds as well and thus
\begin{equation*}
    \sum_{i \in E_A} v_ix_i(E) = \sum_{i \in E_A} v_i \leq \sum_{i \in D_A} v_i = \sum_{i \in D_A} v_ix_i(D),
\end{equation*}
where both equalities follow from part \ref{item-fractional-greedy-size} of \Cref{pro:fractional-greedy} and the inequality from the statement. If $\sum_{i\in E}v_i < C \leq \sum_{i\in D} v_i$, then
\begin{equation*}
    \sum_{i \in E_A} v_ix_i(E)
    = \sum_{i \in E} v_i - \sum_{i \in E \setminus E_A} v_i \leq C - \sum_{i \in E \setminus E_A} v_i 
    \leq C - \sum_{i \in D \setminus D_A} v_ix_i(D) = \sum_{i \in D_A} v_ix_i(D),
\end{equation*}
where the first and last equalities follow from part \ref{item-fractional-greedy-size} of \Cref{pro:fractional-greedy}, and the inequalities from the hypothesis that $\sum_{i \in E} v_i \leq C$ and the fact that $x_i(D)\leq 1$ for every $i\in D$. Finally, if $\sum_{i\in E}v_i \geq C$, then we know from the statement that $\sum_{i\in D}v_i \geq C$ holds as well and we can use part \ref{item-fractional-greedy-size} of \Cref{pro:fractional-greedy} together with the fact that $x_i(E) \geq x_i(D)$ for every $i \in E \setminus E_A$ to obtain
\begin{equation*}
    \sum_{i \in E_A} v_ix_i(E) = C - \sum_{i \in E\setminus E_A} v_ix_i(E) \leq C - \sum_{i \in D \setminus D_A} v_ix_i(D) = \sum_{i \in D_A} v_ix_i(D).\qedhere
\end{equation*}
\end{proof}

We now apply the previous result twice to conclude \Cref{lem:dominance-implication}, once taking $A=\{a\}$ and once taking $A = [n] \setminus \{a\}$.

\begin{proof}[Proof of \Cref{lem:dominance-implication}]
    
Let $E,~ E'$, and $a$ be as in the statement. We let
    \begin{equation*}
        D = E_a \cup (E'\setminus E'_a)
    \end{equation*}
    be a set with the items of $a$ that belong to $E$ and the items of all other agents that belong to $E'$. Note that $E_a = D_a$ and $E' \setminus E'_a = D \setminus D_a$ by definition of $D$. 

    We first relate $E$ and $D$ with $A = [n] \setminus \{a\}$. We have $|E \setminus E_A| = |E_a| = |D_a| = |D \setminus D_A|$ and we have $|E_A| = |E \setminus E_a| \leq |E' \setminus E'_a| = |D \setminus D_a| = |D_A|$ by the assumption of $a$-dominance. Additionally, from the definition of $a$-dominance, we have that for every $k \in [ |E \setminus E_a| ]$ it holds that $v_{i_k(E \setminus E_a)} \leq v_{i_k(E' \setminus E'_a)} = v_{i_k(D \setminus D_a)}$. Therefore, \Cref{lem:larger-sizes-implication} implies that 
    \begin{equation}
        \sum_{i \in E \setminus E_a} v_i x_i(E) \leq \sum_{i \in D \setminus D_a} v_i x_i(D).\label{eq:relation-E-F}
    \end{equation}

    We now relate $E'$ and $D$ with $A = \{a\}$. We have $ |E' \setminus E'_A| = |D \setminus D_A|$ since $A = \{a\}$ and we have $|E'_A| = |E'_a| \leq |E_a| = |D_a| = |D_A|$ by the assumption of $a$-dominance. Again, from the definition of $a$-dominance, for every $k \in [ |E'_a| ]$ we have \mbox{$v_{i_k(E'_a)} \leq v_{i_k(E_a)} = v_{i_k(D_a)}$}. Therefore, \Cref{lem:larger-sizes-implication} implies that
    \begin{equation}
        \sum_{i \in E'_a} v_i x_i(E') \leq \sum_{i \in D_a} v_i x_i(D).\label{eq:relation-F-E'}
    \end{equation}

    If $\sum_{i\in E}v_i <C$, part \ref{item-fractional-greedy-size} of \Cref{pro:fractional-greedy} implies that $x_i(E)=1$ for every $i\in E$, thus inequality \eqref{eq:relation-F-E'} yields
    \begin{equation*}
        \sum_{i \in E_a} v_i x_i(E) = \sum_{i \in D_a} v_i \geq \sum_{i \in D_a} v_i x_i(D) \geq \sum_{i \in E'_a} v_i x_i(E'),
    \end{equation*}
    as claimed. If, on the other hand, $\sum_{i\in E}v_i \geq C$, then the definition of $D$ and the fact that $E$ \mbox{$a$-dominates} $E'$ implies that $\sum_{i\in D}v_i \geq C$ as well. Therefore, applying part \ref{item-fractional-greedy-size} of \Cref{pro:fractional-greedy} together with inequalities \eqref{eq:relation-E-F} and \eqref{eq:relation-F-E'}, we obtain that
    \begin{align*}
        \sum_{i \in E_a} v_i x_i(E) & = C - \sum_{i \in E \setminus E_a} v_i x_i(E) \geq C - \sum_{i \in D \setminus D_a} v_i x_i(D) = \sum_{i \in D_a} v_i x_i(D) \geq \sum_{i \in E'_a} v_i x_i(E'),
    \end{align*}
    as claimed.
\end{proof}

We now state the key property to conclude that \nalgud~with parameter \mbox{$\beta\in [1/2, 2/3]$} is strategyproof: Whenever the mechanism does not return an individual optimum, the notion of $a$-dominance captures the relation between sets $R(E)$ and $R(E\setminus \{i\})$ for any $i \in E_a$.

\begin{lemma}
    \label{lem:dominance-deletion}
    Let $\beta\in [ 1/2, 2/3]$ be arbitrary and let $E$ be a set of items of an instance in $\calI_{\mathrm{UD}}$ such that $v(\OPT(E_a)) < \beta \cdot C$ for every $a \in [n]$. Then, for every $a \in [n]$ and $i \in E_a$, $R(E)$ $a$-dominates $R(E \setminus \{i\})$.
\end{lemma}

Before proving this lemma, we introduce two intermediate results that cover some properties of instances in $\calI_{\mathrm{UD}}$ where the mechanism does not return an individual optimum for some $\beta \in [1/2, 2/3]$.
The first one states that no item in $R(E)$ can be larger than $i^*(E)$.

\begin{lemma}
    \label{lem:i-star-largest}
    Let $E$ be a set of items of an instance in $\calI_{\mathrm{UD}}$ such that $v(\OPT(E_a)) < \beta\cdot C$ for every $a \in [n]$. Then, it holds that $i^*(E) = \arg\max \{v_i \mid i \in R(E)\}$.
\end{lemma}

\begin{proof}
    Let $E$ be as in the statement and suppose, towards a contradiction, that $i^*$ is not the most valuable item in $R$, i.e., $\{i \in R \mid v_i > v_{i^*} \} \not= \emptyset$. Let $i$ be the least valuable item in this set, i.e., $i = \arg\min \{v_i \mid i \in R \text{ s.t.\ } v_i>v_{i^*} \}$. Since $i \in R$, we know that $v_{i^*} + v_i \leq C$.
    
    As $v_i > v_{i^*}$ and $i^*$ is the most valuable item in $P$, it must hold that $i \not\in P$, so there exists $j \in E \setminus E_{a(i)}$ with $v_{i^*} < v_j < v_i$ such that $v_i+v_j > C$. From the definition of $i$ as the least valuable item in $R$ with $v_i > v_{i^*}$, this item $j$ cannot belong to $R$, i.e., $v_{i^*}+v_j>C$. But on the other hand, we know that $v_{i^*}+v_j < v_{i^*}+v_i \leq C$, a contradiction.
\end{proof}

The second result we need before proving \Cref{lem:dominance-deletion} states three conditions under which a pair of items $i,j \in E$ cannot be packed together: (i) $i=i^*(E)$ and $v_j > v_i$, (ii) $i,j \in E_a$ for some agent $a \in [n]$ and $\min\{v_i, v_j\} \geq (1-\beta)C$, (iii) $i,j \in E \setminus R(E)$.

\begin{lemma}
    \label{lem:items-not-fitting}
    Let $\beta\in [1/2, 2/3]$ be arbitrary, let $E$ be a set of items of an instance in $\calI_{\mathrm{UD}}$ such that $v(\OPT(E_a)) < \beta\cdot C$ for every $a \in [n]$, and let $i,j \in E$ be items such that at least one of the following conditions holds:
    \begin{enumerate}[label=(\roman*)]
        \item $i=i^*(E)$ and $v_j > v_i$;\label{item:i-star-not-fitting}
        \item $i,j \in E_a$ for some $a \in [n]$ and $\min\{v_i, v_j\} \geq (1-\beta)C$;\label{item:large-items-not-fitting}
        \item $i,j \in E \setminus R(E)$.\label{item:deleted-items-not-fitting}
    \end{enumerate}
    Then, $v_i+v_j > C$.
\end{lemma}

\begin{proof}
    We let $E$ be as in the statement and prove that each condition implies $v_i+v_j>C$. 

    In order to prove that part \ref{item:i-star-not-fitting} implies $v_i+v_j > C$, we let $i$ and $j$ be such that $i=i^*$ and $v_j>v_{i^*}$. Let also $j' = \arg\min \{v_k \mid k\in E \text{ s.t.\ } v_k>v_i\}$. Since $i = \arg\max \{v_k \mid k \in P\}$, we know that $j' \not\in P$ and, thus, there is an item $k'\in E \setminus E_{a(j')}$ with $v_{j'}>v_{k'}$ and $v_{j'}+v_{k'} > C$. But $i = \arg\max\{v_k \mid k\in E \text{ s.t\ } v_{j'}>v_k\}$, hence $v_i+v_{j'} \geq v_{k'}+v_{j'} > C$. As $v_j\geq v_{j'}$, the result follows.

    We now prove that condition \ref{item:large-items-not-fitting} is sufficient to conclude $v_i+v_j > C$. We consider an agent $a \in [n]$ and items $i,j \in E_a$ such that $v_i \geq (1-\beta)C$ and $v_j \geq (1-\beta)C$. If $v_i+v_j \leq C$, then packing both $i$ and $j$ is a feasible solution, thus 
    \begin{equation*}
        v(\OPT(E_a)) \geq v_i+v_j \geq 2(1-\beta)C \geq \beta\cdot C,
    \end{equation*}
    where the last inequality follows from the fact that $\beta \leq 2/3$.
    This is a contradiction to the statement of the lemma.
    
    To prove that part \ref{item:deleted-items-not-fitting} implies $v_i+v_j > C$, we let $i,j \in E \setminus R$ be arbitrary items not in $R$. The definition of $R$ implies that $v_{i^*} + v_i > C$ and $v_{i^*} + v_j > C$, so if we have $v_i>v_{i^*}$ or $v_j>v_{i^*}$ the result follows directly. We thus assume that $v_i<v_{i^*}$ and $v_j<v_{i^*}$; as $i^* \in P$ we must have that $\{i,j\} \subset E_{a^*}$. On the other hand, $v(\OPT(E_a)) < \beta\cdot C$ implies that $v_{i^*} < \beta\cdot C$ and thus $\min \{v_i,v_j\} > (1-\beta)C$. Therefore, part \ref{item:large-items-not-fitting} implies that $v_i+v_j>C$.
\end{proof}

To prove \Cref{lem:dominance-deletion}, we distinguish four cases according to the hidden item $i$ and its relation to other items: (a) $i \neq i^*(E)$ and $a = a^*(E)$, (b) $i \neq i^*(E)$ and $a \neq a^*(E)$, (c) $i=i^*(E)$ and $v_{i^*(E\setminus \{i\})} > v_i$, and (d) $i=i^*(E)$ and $v_{i^*(E\setminus \{i\})} < v_i$. For each of these cases, we apply the results above to derive the relation between $R(E)$ and $R(E\setminus \{i\})$ and conclude. \Cref{fig:a-dominance} depicts an example of each case.
\begin{figure}[t]
\centering

\begin{tikzpicture}[PatternBlueNE/.style={pattern=north east lines, pattern color=blue!50!black}, PatternOrangeNE/.style={pattern=north east lines, pattern color=orange!50!black}, PatternGreenNE/.style={pattern=north east lines, pattern color=green!50!black}, PatternVioletNE/.style={pattern=north east lines, pattern color=violet!50!black}]


\draw[fill = lightblue] (0,1.5) rectangle (1.2,2.2);
\draw[preaction={fill = lightorange}, PatternOrangeNE] (1.2,1.5) rectangle (2.22,2.2);
\draw[fill = lightorange] (2.22,1.5) rectangle (3.22,2.2);   
\draw[preaction={fill = lightblue}, PatternBlueNE] (3.22,1.5) rectangle (4.12,2.2);
\draw[preaction={fill = lightviolet}, PatternVioletNE] (4.12,1.5) rectangle (4.74,2.2);
\draw[preaction={fill = lightviolet}, PatternVioletNE] (4.74,1.5) rectangle (5.34,2.2);
\draw[preaction={fill = lightgreen}, PatternGreenNE] (5.34,1.5) rectangle (5.74,2.2);
\draw[preaction={fill = lightorange}, PatternOrangeNE] (5.74,1.5) rectangle (5.94,2.2);

\draw [very thick, -](2,1.3) -- (2,2.4);
\Text[x=2,y=2.6]{$C=10$};

\Text[x=-0.3,y=1.85]{$v$};

\Text[x=0.6,y=1.85]{$6$};

\Text[x=1.71,y=1.85]{$5.1$};

\Text[x=2.72,y=1.85]{$5$};

\Text[x=3.67,y=1.85]{$4.5$};

\Text[x=4.43,y=1.85]{$3.1$};

\Text[x=5.04,y=1.85]{$3$};

\Text[x=5.54,y=1.85]{$2$};

\Text[x=5.84,y=1.85]{$1$};

\begin{scope}[shift={(0,-1.2)}]

\draw[fill = lightblue] (0,1.5) rectangle (1.2,2.2);
\draw[preaction={fill = lightorange}, PatternOrangeNE] (1.2,1.5) rectangle (2.22,2.2);  
\draw[preaction={fill = lightblue}, PatternBlueNE] (2.22,1.5) rectangle (3.12,2.2);
\draw[preaction={fill = lightviolet}, PatternVioletNE] (3.12,1.5) rectangle (3.74,2.2);
\draw[preaction={fill = lightviolet}, PatternVioletNE] (3.74,1.5) rectangle (4.34,2.2);
\draw[preaction={fill = lightgreen}, PatternGreenNE] (4.34,1.5) rectangle (4.74,2.2);
\draw[preaction={fill = lightorange}, PatternOrangeNE] (4.74,1.5) rectangle (4.94,2.2);

\draw [very thick, -](2,1.3) -- (2,2.4);

\Text[x=-0.3,y=1.85]{$v$};

\Text[x=0.6,y=1.85]{$6$};

\Text[x=1.71,y=1.85]{$5.1$};

\Text[x=2.67,y=1.85]{$4.5$};

\Text[x=3.43,y=1.85]{$3.1$};

\Text[x=4.04,y=1.85]{$3$};

\Text[x=4.54,y=1.85]{$2$};

\Text[x=4.84,y=1.85]{$1$};

\end{scope}


\draw[fill = lightblue] (7,1.5) rectangle (8.2,2.2);
\draw[fill = lightorange] (8.2,1.5) rectangle (9.3,2.2);
\draw[preaction={fill = lightblue}, PatternBlueNE] (9.3,1.5) rectangle (10.26,2.2);
\draw[preaction={fill = lightviolet}, PatternVioletNE] (10.26,1.5) rectangle (11.06,2.2);
\draw[preaction={fill = lightgreen}, PatternGreenNE] (11.06,1.5) rectangle (11.76,2.2);
\draw[preaction={fill = lightgreen}, PatternGreenNE] (11.76,1.5) rectangle (12.16,2.2);

\draw [very thick, -](9,1.3) -- (9,2.4);
\Text[x=9,y=2.6]{$C=10$};

\Text[x=6.7,y=1.85]{$v$};

\Text[x=7.6,y=1.85]{$6$};

\Text[x=8.75,y=1.85]{$5.5$};

\Text[x=9.78,y=1.85]{$4.8$};

\Text[x=10.66,y=1.85]{$4$};

\Text[x=11.41,y=1.85]{$3.5$};

\Text[x=11.96,y=1.85]{$2$};

\begin{scope}[shift={(0,-1.2)}]

\draw[preaction={fill = lightblue}, PatternBlueNE] (7,1.5) rectangle (8.2,2.2);
\draw[fill = lightblue] (8.2,1.5) rectangle (9.16,2.2);
\draw[preaction={fill = lightviolet}, PatternVioletNE] (9.16,1.5) rectangle (9.96,2.2);
\draw[preaction={fill = lightgreen}, PatternGreenNE] (9.96,1.5) rectangle (10.66,2.2);
\draw[preaction={fill = lightgreen}, PatternGreenNE] (10.66,1.5) rectangle (11.06,2.2);

\draw [very thick, -](9,1.3) -- (9,2.4);

\Text[x=6.7,y=1.85]{$v$};

\Text[x=7.6,y=1.85]{$6$};

\Text[x=8.68,y=1.85]{$4.8$};

\Text[x=9.56,y=1.85]{$4$};

\Text[x=10.31,y=1.85]{$3.5$};

\Text[x=10.86,y=1.85]{$2$};

\end{scope}

\Text[x=2.6, y=-0.4]{(a) $i \neq i^*(E)$ and $a = a^*(E)$};
\Text[x=9.4, y=-0.4]{(b) $i \neq i^*(E)$ and $a \neq a^*(E)$};

\begin{scope}[shift={(0,-4)}]


\draw[fill = lightblue] (0,1.5) rectangle (1.2,2.2);
\draw[preaction={fill = lightorange}, PatternOrangeNE] (1.2,1.5) rectangle (2.22,2.2);
\draw[preaction={fill = lightblue}, PatternBlueNE] (2.22,1.5) rectangle (3.06,2.2);
\draw[preaction={fill = lightgreen}, PatternGreenNE] (3.06,1.5) rectangle (3.76,2.2);
\draw[preaction={fill = lightgreen}, PatternGreenNE] (3.76,1.5) rectangle (4.26,2.2);
\draw[preaction={fill = lightorange}, PatternOrangeNE] (4.26,1.5) rectangle (4.46,2.2);

\draw [very thick, -](2,1.3) -- (2,2.4);
\Text[x=2,y=2.6]{$C=10$};

\Text[x=-0.3,y=1.85]{$v$};

\Text[x=0.6,y=1.85]{$6$};

\Text[x=1.71,y=1.85]{$5.1$};

\Text[x=2.64,y=1.85]{$4.2$};

\Text[x=3.41,y=1.85]{$3.5$};

\Text[x=4.01,y=1.85]{$2.5$};

\Text[x=4.36,y=1.85]{$1$};

\begin{scope}[shift={(0,-1.2)}]

\draw[preaction={fill = lightblue}, PatternBlueNE] (0,1.5) rectangle (1.2,2.2);
\draw[fill = lightblue] (1.2,1.5) rectangle (2.04,2.2);
\draw[preaction={fill = lightgreen}, PatternGreenNE] (2.04,1.5) rectangle (2.74,2.2);
\draw[preaction={fill = lightgreen}, PatternGreenNE] (2.74,1.5) rectangle (3.24,2.2);
\draw[preaction={fill = lightorange}, PatternOrangeNE] (3.24,1.5) rectangle (3.44,2.2);

\draw [very thick, -](2,1.3) -- (2,2.4);

\Text[x=-0.3,y=1.85]{$v$};

\Text[x=0.6,y=1.85]{$6$};

\Text[x=1.62,y=1.85]{$4.2$};

\Text[x=2.39,y=1.85]{$3.5$};

\Text[x=2.99,y=1.85]{$2.5$};

\Text[x=3.34,y=1.85]{$1$};

\end{scope}

\begin{scope}[shift={(7,0)}]

\draw[fill = lightblue] (0,1.5) rectangle (1.2,2.2);
\draw[preaction={fill = lightorange}, PatternOrangeNE] (1.2,1.5) rectangle (2.3,2.2);
\draw[fill = lightorange] (2.3,1.5) rectangle (3.26,2.2);  
\draw[preaction={fill = lightblue}, PatternBlueNE] (3.26,1.5) rectangle (4.1,2.2);
\draw[preaction={fill = lightgreen}, PatternGreenNE] (4.1,1.5) rectangle (4.9,2.2);
\draw[preaction={fill = lightgreen}, PatternGreenNE] (4.9,1.5) rectangle (5.3,2.2);

\draw [very thick, -](2,1.3) -- (2,2.4);
\Text[x=2,y=2.6]{$C=10$};

\Text[x=-0.3,y=1.85]{$v$};

\Text[x=0.6,y=1.85]{$6$};

\Text[x=1.75,y=1.85]{$5.5$};

\Text[x=2.78,y=1.85]{$4.8$};

\Text[x=3.68,y=1.85]{$4.2$};

\Text[x=4.5,y=1.85]{$4$};

\Text[x=5.1,y=1.85]{$2$};

\end{scope}

\begin{scope}[shift={(7,-1.2)}]

\draw[fill = lightblue] (0,1.5) rectangle (1.2,2.2);
\draw[preaction={fill = lightorange}, PatternOrangeNE] (1.2,1.5) rectangle (2.16,2.2);  
\draw[preaction={fill = lightblue}, PatternBlueNE] (2.16,1.5) rectangle (3,2.2);
\draw[preaction={fill = lightgreen}, PatternGreenNE] (3,1.5) rectangle (3.8,2.2);
\draw[preaction={fill = lightgreen}, PatternGreenNE] (3.8,1.5) rectangle (4.2,2.2);

\draw [very thick, -](2,1.3) -- (2,2.4);

\Text[x=-0.3,y=1.85]{$v$};

\Text[x=0.6,y=1.85]{$6$};

\Text[x=1.68,y=1.85]{$4.8$};

\Text[x=2.58,y=1.85]{$4.2$};

\Text[x=3.4,y=1.85]{$4$};

\Text[x=4,y=1.85]{$2$};

\end{scope}

\Text[x=2.6, y=-0.4]{(c) $i=i^*(E)$ and $v_{i^*(E\setminus \{i\})} > v_i$};
\Text[x=9.4, y=-0.4]{(d) $i=i^*(E)$ and $v_{i^*(E\setminus \{i\})} < v_i$};

\end{scope}

\end{tikzpicture}
\caption{Examples of $a$-dominance of the restricted set of items $R(E)$ over $R(E\setminus \{i\})$ for some $i \in E_a$ when running \nalgud~with parameter $\beta=1/\phi \approx 0.618$. Items of the same color belong to the same agent; those of agent $a$ are orange. Items with a background of diagonal lines belong to the restricted set~$R$.  
Four cases are considered as in the proof of \Cref{lem:dominance-deletion}.
In all of them, the values of the orange items in $R$ before the deletion of $i$ are lexicographically larger than those after deletion; the opposite is true for the remaining items.}
\label{fig:a-dominance}
\end{figure}
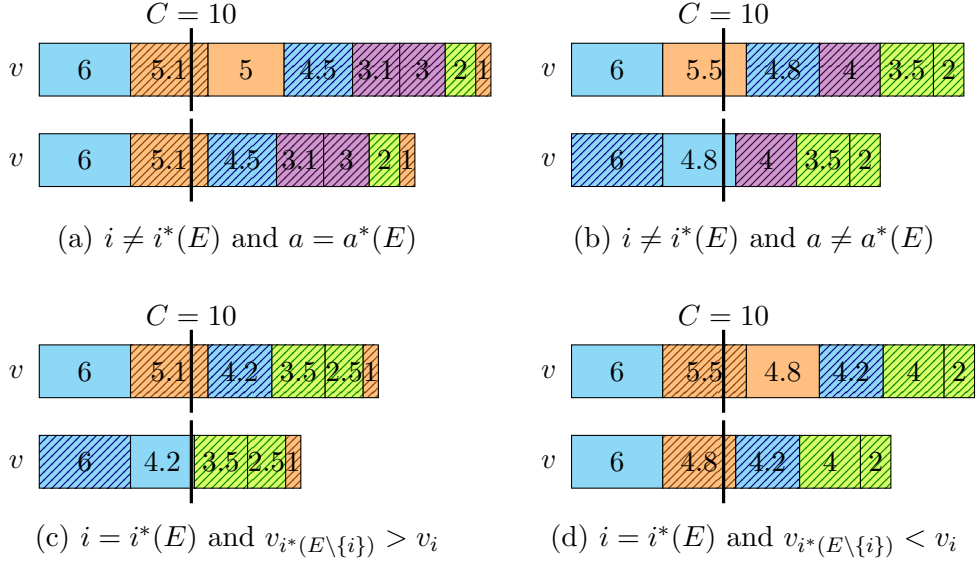

\begin{proof}[Proof of \Cref{lem:dominance-deletion}]

    We consider $\beta$ and $E$ as in the statement, an agent $a \in [n]$, and an item $i \in E_a$. Throughout the proof we will use $R = R(E)$ and $R' = R(E \setminus \{i\})$ for simplicity. We will distinguish four cases according to $i$ and its relation to other items.

    We first suppose that $i \neq i^*(E)$ and $a = a^*(E)$. We claim that $i^*(E \setminus \{i\}) = i^*(E)$. Assume the claim to be wrong. Note that for every item $j \in E_{a} \setminus \{i\}$ we have $j \in P(E) $ if and only if $ j \in P(E \setminus \{i\}) $, because whether or not an item belongs to $P$ depends on other agents' items and they remain the same.  Then there has to be an item $j \in E \setminus E_a$ with $j \in P(E \setminus \{i\}) \setminus P(E)$ and $v_j > v_{i^*(E)}$. But then, by part~\ref{item:i-star-not-fitting} of \Cref{lem:items-not-fitting}, we have $v_j + v_{i^*(E)} > C$ and therefore, $j \notin P(E \setminus \{i\})$, a contradiction. Since the claim is true, $R'= R \setminus \{i\}$, which implies that 
    \begin{equation*}
        R'_a = R_a \setminus \{i\} \text{ and } R' \setminus R'_a = R \setminus R_a.
    \end{equation*}
    We conclude that the sets $R$ and $R'$ satisfy conditions \ref{item:a-dominance-i}-\ref{item:a-dominance-iii} in the definition of $a$-dominance, i.e. $R(E)$ $a$-dominates $R(E\setminus \{i\})$. An example of this case is depicted in part (a) of \Cref{fig:a-dominance}.

    We now suppose that $i \neq i^*(E)$ and $a \neq a^*(E)$. We first claim that $i^*(E \setminus \{i\})$ belongs to agent $a^*(E)$. Note that $i^*(E) \in P(E \setminus \{i\})$ since $v_{i^*(E)}+v_j \leq C$ for every $j \in (E \setminus \{i\}) \setminus E_{a^*(E)}$ with $v_{i^*(E)} > v_j$ due to the definition of $i^*(E)$. Thus, for the claim to be wrong, there has to be an item $j \in (E \setminus \{i\}) \setminus E_{a^*(E)}$ with $j \in P(E \setminus \{i\}) \setminus P(E)$ and $v_j > v_{i^*(E)}$. But then, by part~\ref{item:i-star-not-fitting} of \Cref{lem:items-not-fitting}, we have $v_j + v_{i^*(E)} > C$ and therefore, $j \notin P(E \setminus \{i\})$, a contradiction. Therefore, we obtain that $i^*(E\setminus \{i\}) \in E_{a^*(E)}$ with $v_{i^*(E\setminus \{i\})} \geq v_{i^*(E)}$. If $ i^*(E) = i^*(E \setminus \{i\}) $ we can conclude as in the first case. Otherwise, we have $v_{i^*(E\setminus \{i\})} > v_{i^*(E)}$, so part~\ref{item:i-star-not-fitting} of \Cref{lem:items-not-fitting} yields that $v_{i^*(E\setminus \{i\})} + v_{i^*(E)} > C$ and we conclude that $i^*(E) \notin R'$. Further, we have for every $j \in E \setminus \{i, i^*(E), i^*(E \setminus \{i\})\}$ that $j \in R' $ if and only if $j \in R$. The forward implication is straightforward since $v_{i^*(E\setminus \{i\})} > v_{i^*(E)}$. For the converse, note that if $j \notin R'$, part~\ref{item:deleted-items-not-fitting} of \Cref{lem:items-not-fitting} applied to $E \setminus \{i\} $ yields $v_j + v_{i^*(E)} > C$, thus $j \notin R$. Altogether, $R' = (R \setminus \{i, i^*(E)\}) \cup \{i^*(E\setminus \{i\})\}$, which implies that
    \begin{equation*}
        R'_a = R_a \setminus \{i\} \text{ and } R' \setminus R'_a = ((R \setminus R_a) \setminus \{i^*(E)\}) \cup \{i^*(E\setminus \{i\})\}
    \end{equation*}
    with $v_{i^*(E\setminus \{i\})} \geq v_{i^*(E)}$. We conclude that $R$ and $R'$ satisfy conditions \ref{item:a-dominance-i}-\ref{item:a-dominance-iii} in the definition of $a$-dominance, i.e. $R(E)$ $a$-dominates $R(E\setminus \{i\})$.
    Part (b) of \Cref{fig:a-dominance} shows an example of this case.

    In what follows, we suppose that $i = i^*(E)$ (and thus $a=a^*(E)$). We start by further assuming that $v_{i^*(E\setminus \{i\})} > v_i$. This immediately yields $i^*(E\setminus \{i\}) \in P(E\setminus \{i\}) \setminus P(E)$, hence $a^*(E\setminus \{i\}) \neq a$ since we would have $i^*(E \setminus \{i\}) \in P(E) $ otherwise. For every $j \in E \setminus \{i, i^*(E \setminus \{i\})\}$ with $v_j>v_i$, we have that $v_j + v_{i^*(E\setminus \{i\})} > v_j + v_i > C$, where the last inequality follows from part~\ref{item:i-star-not-fitting} of \Cref{lem:items-not-fitting}. Thus, $j \notin R$ and $j \notin R'$. Similarly, if $j \in E \setminus \{i, i^*(E \setminus \{i\})\}$ is such that $v_j<v_i$ and $j\in R'$, the fact that $j\in R$ follows directly since $v_{i^*(E\setminus \{i\})} > v_i$. Moreover, if $j \in E \setminus \{i, i^*(E \setminus \{i\})\}$ is such that $v_j<v_i$, $j\in R$ and $j\notin E_{a^*(E \setminus \{i\})}$, we obtain that $j\in R'$ since otherwise we would have $i^*(E \setminus \{i\}) \not \in P(E\setminus \{i\})$. Finally, we claim that the set $\{j\in E \setminus \{i, i^*(E \setminus \{i\})\} \mid v_j<v_i \text{ and }j\in R\setminus R'\}$ is either empty or contains a single item from $E_{a^*(E \setminus \{i\})}$. Indeed, if we have two distinct items $j_1$ and $j_2$ in this set, then the fact that they do not belong to $R'$ yields $v_{j_1}\geq (1-\beta)C$ and $v_{j_2}\geq (1-\beta)C$, so part~\ref{item:large-items-not-fitting} of \Cref{lem:items-not-fitting} implies $v_{j_1} + v_{j_2} > C$. However, the fact that they both belong to $R$ and have less value than $i$ implies $v_{j_1}\leq \frac{1}{2}C$ and $v_{j_2}\leq \frac{1}{2}C$, thus $v_{j_1} + v_{j_2} \leq C$, a contradiction. Altogether, if $\{j\in E \setminus \{i, i^*(E \setminus \{i\})\} \mid v_j<v_i \text{ and }j\in R\setminus R'\} = \emptyset$, we have that
    \begin{equation*}
        R'_a = R_a \setminus \{i\} \text{ and } R' \setminus R'_a = (R \setminus R_a) \cup \{i^*(E\setminus \{i\})\}
    \end{equation*}
    and conclude that the sets $R$ and $R'$ satisfy conditions \ref{item:a-dominance-i}-\ref{item:a-dominance-iii} in the definition of $a$-dominance.
    If $\{j\in E \setminus \{i, i^*(E \setminus \{i\})\} \mid v_j<v_i \text{ and }j\in R\setminus R'\} = \{k\}$ with $k\in E_{a^*(E \setminus \{i\})}$, we have that
    \begin{equation*}
        R'_a = R_a \setminus \{i\} \text{ and } R' \setminus R'_a = ((R \setminus R_a) \setminus \{k\}) \cup \{i^*(E\setminus \{i\})\}
    \end{equation*}
    with $v_k < v_{i^*(E\setminus \{i\})}$. Again, the sets $R$ and $R'$ satisfy conditions \ref{item:a-dominance-i}-\ref{item:a-dominance-iii} in the definition of $a$-dominance.
    In either case, $R(E)$ $a$-dominates $R(E\setminus \{i\})$. 
    An example of this case (with an item $k$ as described) is depicted in part (c) of \Cref{fig:a-dominance}. 

    We finally suppose that $i = i^*(E)$ and $v_{i^*(E\setminus \{i\})} < v_{i^*(E)}$. We claim that, in this case, $i^*(E\setminus \{i\}) = \arg\max \{v_i \mid i \in E \text{ s.t.\ } v_i < v_{i^*(E)} \}$. We now conclude the result assuming this claim and prove it afterward. For every $j \in E \setminus \{i, i^*(E\setminus \{i\})\}$, we have that $j \in R $ if and only if $ j \in R'$. 
    To see this, first note that if $j \in E \setminus \{i, i^*(E\setminus \{i\})\}$ is such that $v_j > v_i$, \Cref{lem:i-star-largest} directly implies that $j \notin R$ and $j \notin R'$. On the other hand, if $j \in E \setminus \{i, i^*(E\setminus \{i\})\}$ is such that $v_j < v_{i^*(E\setminus \{i\})}$ and $j \in R$, then $j \in R'$ since $v_{i^*(E\setminus \{i\})} < v_{i^*(E)}$. Finally, let $j \in E \setminus \{i, i^*(E\setminus \{i\})\}$ be such that $v_j < v_{i^*(E\setminus \{i\})}$ and $j \in R'$, i.e., $v_j + v_{i^*(E\setminus \{i\})} \leq C$. Suppose towards a contradiction that $j\not\in R$, i.e., $v_i+v_j>C$. If $j \notin E_a$, this is an immediate contradiction to the fact that $i\in P(E)$. If $i^*(E\setminus \{i\}) \notin E_a$, we obtain $v_i + v_{i^*(E\setminus \{i\})} > v_i + v_j > C$, again a contradiction to the fact that $i\in P(E)$. If $j \in E_a$ and $i^*(E\setminus \{i\}) \in E_a$, then $j \notin R$ implies that $v_j \geq (1-\beta) C$, so $v_{i^*(E\setminus \{i\})} > v_j $ and part~\ref{item:large-items-not-fitting} of \Cref{lem:items-not-fitting} implies that $v_j + v_{i^*(E\setminus \{i\})} > C$, contradicting that $j \in R'$.
    We conclude that $R' = (R \setminus \{i\}) \cup \{i^*(E\setminus \{i\})\}$. If $i^*(E\setminus \{i\}) \in E_a$, this yields
    \begin{equation*}
        R'_a = (R_a \setminus \{i\}) \cup \{i^*(E\setminus \{i\})\} \text{ and } R' \setminus R'_a = R \setminus R_a
    \end{equation*}
    with $v_{i^*(E\setminus \{i\})} \leq v_{i}$. If $i^*(E\setminus \{i\}) \not\in E_a$, we have that
    \begin{equation*}
        R'_a = R_a \setminus \{i\} \text{ and } R' \setminus R'_a = (R \setminus R_a) \cup \{i^*(E\setminus \{i\})\}.
    \end{equation*}
    In both cases, the sets $R$ and $R'$ satisfy conditions \ref{item:a-dominance-i}-\ref{item:a-dominance-iii} in the definition of $a$-dominance, i.e. $R(E)$ $a$-dominates $R(E\setminus \{i\})$. An example of this case is depicted in part (d) of \Cref{fig:a-dominance}.
     
     We now prove the claim. We denote $j^* = \arg\max \{v_i \mid i \in E \text{ s.t.\ } v_i < v_{i^*(E)} \}$ and $b=a(j^*)$. If $b=a$, we have that for any $j\in E\setminus E_a$ with $v_{j^*} > v_j$, it holds that $v_{j^*} + v_j < v_{i^*(E)} + v_j \leq C$, where the second inequality follows from the fact that $i^*(E) \in P(E)$. If, on the other hand, $b\not= a$, we have that for any $j\in E\setminus E_b$ with $v_{j^*} > v_j$, it holds that $v_{j^*} + v_j < v_{j^*} + v_{i^*(E)} \leq C$, where, once again, the second inequality follows from the fact that $i^*(E) \in P(E)$. Thus, in any case it holds that $j^* \in P(E\setminus \{i\})$ and this item is, by the initial assumption that $v_{i^*(E\setminus \{i\})} < v_{i^*(E)}$, the largest item in this set. This concludes the proof of the claim.
\end{proof}

The strategyproofness of \nalgud\ follows by combining \Cref{lem:dominance-implication,lem:dominance-deletion}. Regarding its approximation guarantee, $\beta$-approximation is trivially attained when a single agent provides a $\beta$-approximation and the mechanism returns their optimum. Otherwise, we crucially rely on the fact that no agent can provide a $\beta$-approximation on their own and that, as a consequence, no item is larger than $\beta \cdot C$. The proof then mainly hinges on two properties: (i) no pair of items in $E\setminus R(E)$ fit together, and (ii) $i^*$ provides a $\smash{\frac{1-\beta}{\beta}}$-approximation of any deleted item: In the worst case, $v_{i^*}=(1-\beta+\varepsilon)C$ and $v_k=(\beta-\varepsilon')C$ for some deleted item $k$ and some small values $\varepsilon>\varepsilon'>0$.
Before proceeding with the proof, we address two specific cases where \nalgrspgreedy\ is easily seen to be at least $\frac{2}{3}$-approximate. When the restricted set $Q\subseteq E$ over which it runs is such that $v(Q) \leq C$, we have $v(\algrspgreedy(E,Q)) \geq v(Q)$ as each agent can pack items with a value of their items contained in $Q$.
When the two most valuable items in $Q$ fit together and $v(Q) > C$, a $\smash{\frac{2}{3}}$-approximation also holds since each of the packed items, which are at least two, has more value than the fractional item and all together give a bound on $\OPT(E)$.

\begin{lemma}
\label{lem:r(s)-small}
    If $E$ is a set of items of an instance in $\calI_{\mathrm{UD}}$ and $Q\subseteq E$ is such that $v(Q)\leq C$, then $v(\algrspgreedy(E,Q)) \geq v(Q)$.
\end{lemma}

\begin{proof}
    Let $E$ and $Q\subseteq E$ be such that $v(Q)\leq C$. This implies that $x_i(Q) = 1$ for every $i \in Q$. Therefore,
    \begin{align*}
        v(\algrspgreedy(E,Q)) & = \sum_{a \in [n]} \max \bigg\{v(E') ~\bigg|~ E' \subseteq E_a \text{ s.t. } v(E') \leq \sum_{i \in Q_a} v_ix_i(Q)\bigg\} \geq \sum_{a \in [n]} v(Q_a) = v(Q),
    \end{align*}
    where the inequality follows from the fact that $v(Q_a) \leq \sum_{i \in Q_a} v_i x_i(Q)$.
\end{proof}

\begin{lemma}
    \label{lem:r(s)-large}
    Let $E$ be a set of items of an instance in $\calI_{\mathrm{UD}}$ and $Q\subseteq E$ be such that $\ell(Q)\geq 2$ and $v(Q) > C$. Let also $v^* = \max\{ v_i \mid i\in Q\}$ denote the maximum value of an item in $Q$. Then,
    \begin{equation*}
        v(\algrspgreedy(E,Q)) \geq \max \bigg\{ 1-\frac{v^*}{C}, \frac{1}{2}+\frac{v^*}{2C} \bigg\} \cdot v(\OPT(E)).
    \end{equation*}
\end{lemma}

\begin{proof}
    Let $E$ and $Q\subseteq E$ be such that $v(Q)>C$. We further let $x = x(Q)$, $\ell = \ell(Q)$ and $i_k = i_k(Q)$ for every $k \in [|Q|]$. We know that
    \begin{equation}
        \sum_{k=1}^{\ell} v_{i_k} = \sum_{k=1}^{|Q|} v_{i_k} x_{i_k} - v_{i_{\ell+1}} x_{i_{\ell+1}} = C - v_{i_{\ell+1}}x_{i_{\ell+1}}.\label{eq:lem-phi-approximation}
    \end{equation}
    We now bound $v_{i_{\ell+1}}x_{i_{\ell+1}}$ from above. On the one hand, it is trivial that this expression is strictly smaller than $v_{i_{\ell+1}}$, which is, in turn, smaller than $v^*$. On the other hand, since $\ell \geq 2$ and $v_{i_1} > v_{i_2} > \cdots > v_{i_{\ell+1}}$, we obtain that 
    \begin{equation*}
        v_{i_{\ell+1}}x_{i_{\ell+1}} \leq C-v^* - \sum_{k=2}^{\ell} v_{i_k} < C-v^* - v_{i_{\ell+1}} < C-v^* - v_{i_{\ell+1}}x_{i_{\ell+1}},
    \end{equation*}
    hence
    \begin{equation*}
        v_{i_{\ell+1}}x_{i_{\ell+1}} < \frac{C-v^*}{2}.
    \end{equation*}
    Plugging these two bounds on $v_{i_{\ell+1}}x_{i_{\ell+1}}$ into \eqref{eq:lem-phi-approximation}, we obtain
    \begin{equation*}
        \sum_{k=1}^{\ell} v_{i_k} > \max\bigg\{C-v^*, \frac{C+v^*}{2} \bigg\}.
    \end{equation*}
    Since $x_{i_k}=1$ for every $k \in [\ell]$, we conclude that
    \begin{align*}
        v(\algrspgreedy(E,Q)) & = \sum_{a \in [n]} \max \bigg\{v(E') ~\bigg|~ E' \subseteq E_a \text{ s.t. } v(E') \leq \sum_{i \in Q_a} v_ix_i\bigg\} \nonumber \\
        & \geq \sum_{a \in [n]} \sum_{k \in [\ell]: i_k \in E_a} v_{i_k} = \sum_{k=1}^{\ell} v_{i_k} \geq \max\bigg\{C-v^*, \frac{C+v^*}{2} \bigg\}.
    \end{align*}
    The inequality in the statement follows directly from this one and from $v(\OPT(E))\leq C$.
\end{proof}

We now proceed with the proof of \Cref{lem:alg-ud-param}. 

\begin{proof}[Proof of \Cref{lem:alg-ud-param}]
    Let $\beta \in [1/2, 2/3]$ be arbitrary. To prove strategyproofness, we fix a set of items $E$, an agent $a \in [n]$, and an item $i \in E_a$ arbitrarily. 
        
    If $a = \arg\max \{v(\OPT(E_b)) \mid b \in [n] \}$ and $v(\OPT(E_a)) \geq \beta\cdot C$, then from the sole fact that $\algud_{\beta}$ returns a feasible set of items we have that
    \begin{equation*}
        v_a(\algud_{\beta}(E \setminus \{i\})) \leq \max\{v(E') \mid E'\subseteq E_a\setminus \{i\} \text{ s.t. } v(E') \leq C\} \leq v(\OPT(E_a)) = v_a(\algud_{\beta}(E)).
    \end{equation*}
    Observe that, if $a = \arg\max \{v(\OPT(E_b \setminus \{i\})) \mid b \in [n] \}$ and $v(\OPT(E_a \setminus \{i\})) \geq \beta\cdot C$, it also holds that $a = \arg\max \{v(\OPT(E_b)) \mid b \in [n] \}$ and $v(\OPT(E_a)) \geq \beta\cdot C$, thus the above chain of inequalities remains valid.
    
    If $a \not= \arg\max \{v(\OPT(E_b \setminus \{i\})) \mid b \in [n] \}$ and $v(\OPT(E_b)) \geq \beta\cdot C$ for some $b \in [n]\setminus \{a\}$, then $v_a(\algud_\beta(E \setminus \{i\})) = 0 \leq v_a(\algud_\beta(E))$.
    Finally, if $a \not= \arg\max \{v(\OPT(E_b \setminus \{i\})) \mid b \in [n] \}$ and $v(\OPT(E_b)) < \beta\cdot C$ for every $b \in [n]\setminus \{a\}$, then
    \begin{align*}
        v_a(\algud_\beta(E \setminus \{i\})) & = \max\bigg\{v(E')~\bigg|~E'\subseteq E_a \text{ s.t. } v(E') \leq \sum_{j \in R_a(E \setminus \{i\})} v_jx_j(R(E \setminus \{i\}))\bigg\} \\
        & \leq \max\bigg\{v(E')~\bigg|~E'\subseteq E_a \text{ s.t. } v(E') \leq \sum_{j \in R_a(E)} v_jx_j(R(E)) \bigg\}\\
        & = v_a(\algud_\beta(E)),
    \end{align*}
    where the inequality follows from \Cref{lem:dominance-implication} and \Cref{lem:dominance-deletion}.
    
    We conclude that, in all the cases above, $v_a(\algud_\beta(E \setminus \{i\})) \leq v_a(\algud_\beta(E))$. \Cref{lem:one-item-suffices} thus implies that \nalgud~with parameter $\beta$ is strategyproof.
    
    We now proceed to show that \nalgud~with parameter $\beta$ is $\smash{\min\big\{ \beta, \frac{1-\beta}{\beta}\big\}}$-approximate, and give conditions whenever it is not $\beta$-approximate.
    If $v(\OPT(E_a)) \geq \beta\cdot C$ for some $a \in [n]$, $\beta$-approximation is trivial; we assume this is not the case in what follows. 
    
    We first tackle some cases in which an approximation factor of $\beta$ follows easily from previous results. If $v(E)\leq C$, it is straightforward that $i^* = \arg\max\{v_i \mid i \in E\}$ and $R = E$, thus
    \begin{equation*}
        v(\algud_\beta(E)) = v(\algrspgreedy(E,R)) \geq v(R) = v(E) = v(\OPT(E)),
    \end{equation*}
    where the inequality follows from \Cref{lem:r(s)-small}.
    Similarly, if $\beta\cdot C \leq v(R) \leq C$, we have that
    \begin{equation*}
        v(\algud_\beta(E)) = v(\algrspgreedy(E,R)) \geq v(R) \geq \beta\cdot C \geq \beta\cdot v(\OPT(E)),
    \end{equation*}
    where the first inequality follows from \Cref{lem:r(s)-small} and the last one from the fact that $v(\OPT(E)) \leq C$. 
    Furthermore, if $v(R) > C$, we have that
    \begin{equation*}
        \frac{v(\algud_\beta(E))}{v(\OPT(E))} = \frac{v(\algrspgreedy(E,R))}{v(\OPT(E))} \geq \max\bigg\{ 1-\frac{v_{i^*}}{C}, \frac{1}{2}+\frac{v_{i^*}}{2C} \bigg\} \geq \frac{2}{3} \geq \beta,
    \end{equation*}
    where the first inequality comes from \Cref{lem:r(s)-large} and the second one from taking the value of $v_{i^*}/C$ that minimizes the maximum, which is $1/3$.
    
    In what follows, we let $E$ be such that $v(E)> C$ and $v(R) < \beta\cdot C$. Note that, in particular, this implies that $E \setminus R \not= \emptyset$. We distinguish two cases, depending on whether $i^*$ belongs to $\OPT(E)$ or not. If $i^* \in \OPT(E)$, we have that $\OPT(E) \subseteq \{i^*\} \cup \{i\in E \mid v_{i^*}+v_i \leq C\} = R$,
    thus \Cref{lem:r(s)-small} implies that
    \begin{equation*}
        v(\algud_\beta(E)) = v(\algrspgreedy(E,R)) \geq v(R) \geq v(\OPT(E)).
    \end{equation*}
    If $i^* \not\in \OPT(E)$, on the other hand, the following chain of inequalities holds:
    \begin{align*}
        \frac{v(\algud_\beta(E))}{v(\OPT(E))} & = \frac{v(\algrspgreedy(E,R))}{v(\OPT(E))} & v(\OPT(E_a)) < \beta\cdot C \text{ for every } a\in [n] \\
        & \geq \frac{v(R)}{v(\OPT(E))} & \text{\Cref{lem:r(s)-small}}\\
        & \geq \frac{v(R)-v(\OPT(E) \cap R)}{v(\OPT(E))-v(\OPT(E) \cap R)} & \frac{v(R)}{v(\OPT(E))} \leq 1\\
        & = \frac{v(R \setminus \OPT(E))}{v(\OPT(E) \setminus R)}\\
        & \geq \frac{v_{i^*}}{v(\OPT(E) \setminus R)} & i^* \in R \setminus \OPT(E) \\
        & \geq \frac{v_{i^*}}{\max\{v_i \mid i\in E\}} & |\OPT(E) \setminus R| \leq 1 \\
        & \geq \frac{1-\beta}{\beta} & \frac{v_{i^*}}{v_i} \geq \frac{1-\beta}{\beta} \text{ for every } i \in E.
    \end{align*}
    Indeed, all steps follow directly from hypotheses, previous results, or direct computations, except for the last two. The second-to-last one comes from the fact that $|\OPT(E) \setminus R| \leq 1$: This is because $v_i+v_j \leq C$ for every $i,j \in \OPT(E)$ and we know from part \ref{item:deleted-items-not-fitting} of \Cref{lem:items-not-fitting} that $v_i+v_j>C$ for every $i,j \in E\setminus R$. The last one comes from two facts. On the one hand, we know that $v_i < \beta\cdot C$ for every $i\in E$. On the other hand, since $E \setminus R \neq \emptyset$, there is an item $k \in E \setminus R$ with $v_k < \beta \cdot C$ and $v_{i^*} + v_k > C $ from part \ref{item:i-star-not-fitting} of \Cref{lem:items-not-fitting}, hence $v_{i^*} \geq (1-\beta) C$. We conclude that, for every $i\in E$, $v_{i^*}/v_i \geq (1-\beta)/\beta$.
    
    Observe that, from the previous analysis, in all cases we have that $v(\algud_\beta(E)) \geq \beta\cdot v(\OPT(E))$ except for the last case when $v(R)<\beta\cdot C$ and $v_{i^*} < \beta\cdot \max\{v_i \mid i\in E\}$. This concludes the proof of the lemma.
\end{proof}

We now turn to obtain specific bounds for deterministic and randomized mechanisms based on the family of mechanisms we have studied. In \Cref{subsec:unit-density-det}, we directly apply \Cref{lem:alg-ud-param} to show the existence of a $\frac{1}{\phi}$-approximate deterministic mechanism and show the tightness of this upper bound for deterministic mechanisms. In \Cref{subsec:unit-density-rand}, we randomly combine \nalgud~with another strategyproof mechanism and beat this approximation factor.

\subsection{Deterministic mechanisms}\label{subsec:unit-density-det}

By choosing $\beta=1/\phi$ in the definition of \nalgud, the following result follows directly from \Cref{lem:alg-ud-param}.

\begin{theorem}
    \label{theo:unit-density-det-UB}
    There is a strategyproof and $\frac{1}{\phi}$-approximate deterministic mechanism on $\calI_{\mathrm{UD}}$.
\end{theorem}

\begin{proof}
The result follows directly from \Cref{lem:alg-ud-param}: \nalgud~with parameter $\beta=1/\phi$ is strategyproof and $\alpha$-approximate for
\begin{equation*}
    \alpha=\min\Bigg\{\frac{1}{\phi}, \frac{1-\frac{1}{\phi}}{\frac{1}{\phi}}\Bigg\} = \frac{1}{\phi}.\qedhere
\end{equation*}
\end{proof}

The following upper bound shows that \nalgud~with parameter $\beta=1/\phi$ provides the best-possible approximation factor among deterministic mechanisms on $\calI_{\mathrm{UD}}$. We obtain the bound by adapting the one by \citet{FeigenbaumJ17} to the case of items with unit density.

\begin{theorem}
    \label{theo:unit-density-det-LB}
    Let $\scM$ be a strategyproof and $\alpha$-approximate deterministic mechanism on $\calI_{\mathrm{UD}}$. Then, $\alpha \leq 1/\phi$.
\end{theorem}

To prove this bound, we consider $n=2$ agents, sets of items $E_1=\{i,j\},~E'_1=\{i\}$, and $E_2=\{k\}$, values $v_i=\phi,~ v_j=\phi-\varepsilon$, and $v_k=1$, and a capacity $C=\phi+1-\varepsilon$ for a sufficiently small $\varepsilon > 0$. The bound follows by studying the value of the packed items of agent $1$ prior and upon the deletion of item $j$.

\begin{proof}[Proof of \Cref{theo:unit-density-det-LB}]
    Suppose that $\scM$ is a strategyproof and \mbox{$\alpha$-approximate} mechanism for $\alpha> 1/\phi$, say $\alpha=1/(\phi-\delta)$ for some $\delta >0$.

    Let $\varepsilon \in (0, \delta \phi)$ be an arbitrary value in this range. We consider $n=2$ agents and sets of items $E_1=\{i,j\},~E'_1=\{i\}$, and $E_2=\{k\}$ with values $v_i=\phi,~ v_j=\phi-\varepsilon$, and $v_k=1$. We fix a capacity  $C=\phi+1-\varepsilon$ and consider two instances: $E=E_1\cup E_2$ and $E'=E'_1\cup E_2$. 

    For the instance $E$, we have that $\OPT(E)=\{j,k\}$ and thus $v(\OPT(E))=\phi+1-\varepsilon$. By $\varepsilon < \delta \phi$ we get
    \begin{equation*}
        \frac{\phi}{\phi+1-\varepsilon} < \frac{1}{\phi-\delta}
    \end{equation*}
    and therefore $\scM(E)=\{j,k\}$ as well, since $\scM$ is $1/(\phi-\delta)$-approximate.
    
    For the instance $E'$, we have that $\OPT(E')=\{i\}$ and thus $v(\OPT(E'))=\phi$. Since $\scM$ is strategyproof and $E'_1 \subseteq E_1$, it must hold that $v_1(\scM(E')) \leq v_1(\scM(E)) = \phi-\varepsilon$ and thus $i\not \in \scM(E')$. But this yields $\scM(E') \in \{ \{k\}, \emptyset\}$, thus
    \begin{equation*}
        \frac{v(\scM(E'))}{v(\OPT(E'))} \leq \frac{1}{\phi},
    \end{equation*}
    a contradiction to the fact that $\scM$ is $\frac{1}{\phi-\delta}$-approximate.
\end{proof}

\subsection{Randomized mechanisms}\label{subsec:unit-density-rand}

In this section, we show that randomization allows us to improve the best-possible approximation factor that deterministic mechanisms can achieve in the unit-density setting. The following is our main result in this regard.

\begin{theorem}
    \label{theo:unit-density-rand-UB}
    There exists a strategyproof and $\smash{\frac{2}{3}}$-approximate randomized mechanism on $\calI_{\mathrm{UD}}$.
\end{theorem}

The mechanism providing this approximation factor, which we refer to as \nalgmixud~and describe in \Cref{alg:mix-ud}, randomizes over two deterministic mechanisms: It returns \nalgud~with parameter $\beta=2/3$ with a probability of $2/3$ and a mechanism we call \nalglarge~with the remaining probability of $1/3$. We denote its output by $\algmixud(E)$.
\begin{figure}[t]
\begin{minipage}[t]{0.48\textwidth}
    \begin{algorithm}[H]
    \caption{\nalgmixud~$(\algmixud)$}\label{alg:mix-ud}
    \begin{algorithmic}
    \Require set of items $E$ (partition $E_1,\ldots,E_n$, values $v$, and sizes $s$), capacity $C$
    \Ensure subset $E^* \subseteq E$ with $s(E^*) \leq C$
    \State $X \gets \text{Bernoulli}(2/3)$
    \If{$X = 1$}
        \State \textbf{return } $\algud_{2/3}(E)$
    \Else
        \State \textbf{return } $\alglarge(E)$
    \EndIf
    \end{algorithmic}
    \end{algorithm}
\end{minipage}
\hfill
\begin{minipage}[t]{0.48\textwidth}
    \begin{algorithm}[H]
    \caption{\nalglarge~$(\alglarge)$}\label{alg:large}
    \begin{algorithmic}
    \Require set of items $E$ (partition $E_1,\ldots,E_n$, values $v$, and sizes $s$), capacity $C$
    \Ensure subset $E^* \subseteq E$ with $s(E^*) \leq C$
    \If{$\max_{a \in [n]} v(\OPT(E_a)) \geq \frac{2}{3} C$}
        \State $a' \gets \arg\max_{a \in [n]} v(\OPT(E_a))$
        \State \textbf{return} $\OPT(E_{a'})$
    \EndIf
    \State $j^* \gets \arg\max \{v_i \mid i \in E \}$, $b^*\gets a(j^*)$
    \State $S \gets E_{b^*} \cup \{i \in E \setminus E_{b^*} \mid v_{j^*} + v_i \leq C\}$
    \State \textbf{return } $\algrspgreedy(E, S)$
    \end{algorithmic}
    \end{algorithm}
\end{minipage}
\end{figure}

\nalglarge~is formally described in \Cref{alg:large}, and we denote its output for a set of items $E$ as $\alglarge(E)$. This mechanism is similar to \nalgud~with parameter $2/3$, but it is tailored to perform better in situations in which the optimum contains a highly valuable item, in which \nalgud~with parameter $2/3$ may be only $\smash{\frac{1}{2}}$-approximate. In the first step, it also checks whether an individual optimum provides a $\smash{\frac{2}{3}}$-approximation directly, returning the most valuable individual optimum if this is the case. Otherwise, it defines $j^*$ as the most valuable item, $b^*$ as its owner, and $S$ as the set containing the items in $E\setminus E_{b^*}$ that fit together with $j^*$, and all items in $E_{b^*}$. It finally returns $\algrspgreedy(E,S)$. \Cref{fig:fit-comparison} depicts the sets $R$ and $S$ defined in \nalgud~and \nalglarge~for two example instances.
We let $S(E)$ denote the set $S$ defined in \nalglarge~when run with input $E$, and let $j^*(E)$ and $b^*(E)$ denote the values of $j^*$ and $b^*$ when running this mechanism with input $E$, respectively. We omit $E$ when clear from the context.

\begin{figure}[t]
\centering
\begin{tikzpicture}[PatternBlueNE/.style={pattern=north east lines, pattern color=blue!50!black}, PatternOrangeNE/.style={pattern=north east lines, pattern color=orange!50!black}, PatternGreenNE/.style={pattern=north east lines, pattern color=green!50!black}, PatternBlueSE/.style={pattern=south east lines, pattern color=blue!50!black}, PatternOrangeSE/.style={pattern=south east lines, pattern color=orange!50!black}, PatternGreenSE/.style={pattern=south east lines, pattern color=green!50!black}]

\draw[preaction={fill = lightblue}, PatternBlueSE] (0,0) rectangle (1.2,0.7);
\draw[preaction={fill = lightblue},PatternBlueSE] (1.2,0) rectangle (2.3,0.7);
\draw[preaction={fill = lightorange}, PatternOrangeNE] (2.3,0) rectangle (3.3,0.7); 
\draw[preaction={fill = lightgreen}, PatternGreenNE] (3.3,0) rectangle (4.1,0.7);
\fill[PatternGreenSE] (3.307,0.007) rectangle (4.093,0.693);
\draw[preaction={fill = lightgreen}, PatternGreenNE] (4.1,0) rectangle (4.6,0.7);
\fill[PatternGreenSE] (4.107,0.007) rectangle (4.593,0.693);
\draw[preaction={fill = lightorange}, PatternOrangeNE] (4.6,0) rectangle (4.8,0.7);
\fill[PatternOrangeSE] (4.607,0.007) rectangle (4.793,0.693);

\draw [very thick, -](2,-0.2) -- (2,0.9);
\Text[x=2,y=1.1]{$C=10$};

\Text[x=-0.3,y=0.35]{$v$};

\Text[x=0.6,y=0.35]{$6$};

\Text[x=1.75,y=0.35]{$5.5$};

\Text[x=2.8,y=0.35]{$5$};

\Text[x=3.7,y=0.35]{$4$};

\Text[x=4.35,y=0.35]{$2.5$};

\Text[x=4.7,y=0.35]{$1$};

\begin{scope}[shift={(7,0)}]

\draw[preaction={fill = lightorange},PatternOrangeSE] (0,0) rectangle (1.3,0.7);
\draw[preaction={fill = lightgreen},PatternGreenNE] (1.3,0) rectangle (2.1,0.7);
\draw[preaction={fill = lightblue},PatternBlueNE] (2.1,0) rectangle (2.3,0.7);   
\fill[PatternBlueSE] (2.107,0.007) rectangle (2.293,0.693);   

\draw [very thick, -](2,-0.2) -- (2,0.9);
\Text[x=2,y=1.1]{$C=10$};

\Text[x=-0.3,y=0.35]{$v$};

\Text[x=0.65,y=0.35]{$6.5$};

\Text[x=1.7,y=0.35]{$4$};

\Text[x=2.2,y=0.35]{$1$};

\end{scope}

\end{tikzpicture}
\caption{Examples of restricted sets $R(E)$ and $S(E)$ when running \nalgud\ with parameter $\beta=2/3$ and \nalglarge, respectively. Items among $E$ with a background of upward diagonal lines belong to $R(E)$; items among $E$ with a background of downward diagonal lines belong to $S(E)$.}
\label{fig:fit-comparison}
\end{figure}
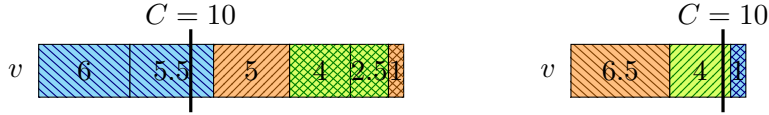

We now proceed with the proof. Strategyproofness of \nalgmixud~follows from the---already proven---fact that \nalgud~is strategyproof and the fact that \nalglarge~is strategyproof as well. To see the latter, we first note that if an agent can provide a $\frac{2}{3}$-approximation on their own, hiding items would never be profitable. Otherwise, we show that hiding items cannot lead to an increase in the quota assigned to each agent by \nalgrspgreedy~using the fact that, when agent $a$ hides item $i$, the set of items of this agent in the restricted set $S$ becomes a subset of the original one, whereas the set of items of all other agents in the restricted set $S$ becomes a superset of the original one. We state this important property in the following lemma.

\begin{lemma}
\label{lem:subsets-large}
    Let $E$ be a set of items of an instance in $\calI_{\mathrm{UD}}$ such that $v(\OPT(E_a)) < \frac{2}{3} C$ for every $a\in [n]$. Then, for every $a \in [n]$ and $i\in E_a$, $S(E\setminus \{i\}) \cap E_a \subseteq S(E) \cap E_a$ and $S(E)\setminus E_a \subseteq S(E \setminus \{i\}) \setminus E_a$.
\end{lemma}

\begin{proof}
Let $E$, $a\in [n]$, and $i\in E_a$ be arbitrary. We use $S=S(E)$ and $S'=S(E\setminus \{i\})$ for simplicity. If $E=\{i\}$, then $S=\{i\}$ and $S' = \emptyset$, so the result follows trivially. We thus assume this is not the case in what follows, so we let $j_2=\arg\max \{v_j \mid j \in E\setminus \{j^*(E)\}\}$ denote the second largest item. We distinguish two cases and conclude the properties in the statement in both of them.

If $i \not= j^*(E)$, it is clear that $j^*(E\setminus \{i\}) = j^*(E)$, thus $b^*(E\setminus \{i\}) = b^*(E)$. We denote these common values simply as $j^*$ and $b^*$. Let $j\in S'_a$ be arbitrary. We have either $a=b^*$ or $j\in E\setminus E_{b^*}$ with $v_j+v_{j^*} \leq C$. In both cases, the fact that $j\in S_a$ directly follows. Let now $j\in S\setminus S_a$ be arbitrary. We have either $j\in E_{b^*}$ or $j\in E\setminus E_{b^*}$ with $v_j+v_{j^*}\leq C$. In either case, the fact that $j\in S'\setminus S'_a$ directly follows. We conclude that $S'_a\subseteq S_a$ and $S\setminus S_a \subseteq S'\setminus S'_a$.

If $i = j^*(E)$, we have that $b^*(E)=a$ and $j^*(E\setminus \{i\})=j_2$; we denote $b_2=b^*(E \setminus \{i\})$ for simplicity. Since $S_a=E_a$, the fact that $S'_a \subseteq S_a$ is straightforward. To see the other inclusion, let $j\in S\setminus S_a$ be an arbitrary item. Since  $j\in S\setminus S_a$, we know that $v_j+v_i\leq C$. As $v_{j_2}<v_i$, this yields $v_j+v_{j_2}\leq C$ and thus $j\in S'\setminus S'_a$. We conclude that $S\setminus S_a \subseteq S'\setminus S'_a$ holds as well.
\end{proof}

The previous lemma immediately implies that $S(E)$ $a$-dominates $S(E\setminus \{i\})$, which is the main ingredient to prove the strategyproofness of \nalglarge.

\begin{lemma}\label{lem:sp-large}
    \nalglarge~is strategyproof on $\calI_{\mathrm{UD}}$.
\end{lemma}

\begin{proof}
    We fix a set of items $E$ of an instance in $\calI_{\mathrm{UD}}$, an agent $a \in [n]$, and an item $i \in E_a$ arbitrarily. If $a = \arg\max \{v(\OPT(E_b)) \mid b \in [n] \}$ and $v(\OPT(E_a)) \geq \frac{2}{3} C$, the sole fact that \nalglarge~returns a feasible set of items implies that
    \begin{equation*}
        v_a(\alglarge(E \setminus \{i\})) \leq \max\{v(E') \mid E'\subseteq E_a\setminus \{i\} \text{ s.t. } v(E') \leq C\} \leq v(\OPT(E_a)) = v_a(\alglarge(E)).
    \end{equation*}
    If $a = \arg\max \{v(\OPT(E_b \setminus \{i\})) \mid b \in [n] \}$ and $v(\OPT(E_a \setminus \{i\})) \geq \frac{2}{3} C$, we necessarily have that $a = \arg\max \{v(\OPT(E_b)) \mid b \in [n] \}$ and $v(\OPT(E_a)) \geq \frac{2}{3} C$ as well, so the previous chain of inequalities remains valid.

    If $a \not= \arg\max \{v(\OPT(E_b \setminus \{i\})) \mid b \in [n] \}$ and $v(\OPT(E_b)) \geq \frac{2}{3} C$ for some $b \in [n]\setminus \{a\}$, then
    \begin{equation*}
        v_a(\alglarge(E \setminus \{i\})) = 0 \leq v_a(\alglarge(E)).
    \end{equation*}
    
    In what follows, we consider the case $v(\OPT(E_b)) < \frac{2}{3} C$ for every $b \in [n]$.

    From \Cref{lem:subsets-large}, we have $S(E\setminus \{i\}) \cap E_a \subseteq S(E) \cap E_a$ and \mbox{$S(E)\setminus E_a \subseteq S(E \setminus \{i\}) \setminus E_a$}. The first inequality yields $|S(E)\cap E_a| \geq |S(E\setminus \{i\}) \cap E_a|$ and $v_{i_k(S(E)\cap E_a)} \geq v_{i_k(S(E\setminus \{i\}) \cap E_a)}$ for every $k\in [|S(E\setminus \{i\}) \cap E_a|]$. The second inequality yields $|S(E)\setminus E_a| \leq |S(E \setminus \{i\}) \setminus E_a|$ and $v_{i_k(S(E)\setminus E_a)} \leq v_{i_k(S(E \setminus \{i\}) \setminus E_a)}$ for every $k \in [|S(E)\setminus E_a|]$. We conclude that $S(E)$ $a$-dominates $S(E\setminus \{i\})$.
    Therefore,
    \begin{align*}
        v_a(\alglarge(E \setminus \{i\})) & = \max\bigg\{v(E')~\bigg|~E'\subseteq E_a \text{ s.t. } v(E') \leq \sum_{j \in S_a(E \setminus \{i\})} v_jx_j(S(E \setminus \{i\}))\bigg\} \\
        & \leq \max\bigg\{v(E')~\bigg|~E'\subseteq E_a \text{ s.t. } v(E') \leq \sum_{j \in S_a(E)} v_jx_j(S(E)) \bigg\}\\
        & = v_a(\alglarge(E)),
    \end{align*}
    where the inequality follows from \Cref{lem:dominance-implication}.
    
    We conclude that, in all the cases above, $v_a(\alglarge(E \setminus \{i\})) \leq v_a(\alglarge(E))$. \Cref{lem:one-item-suffices} thus implies that \nalglarge~is strategyproof.
\end{proof}

Regarding the approximation factor of \nalgmixud\, it will follow from the fact that, whenever one of the two mechanisms over which \nalgmixud\ randomizes is not $\smash{\frac{2}{3}}$-approximate, the other one compensates by providing a better-than-$\smash{\frac{2}{3}}$ approximation, in a way such that the expected value of the packed items provides the desired approximation. Intuitively, if $\algud_{2/3}(E)<\frac{2}{3}\OPT(E)$, the optimum includes the most valuable item and \nalglarge\ performs well; if $\alglarge(E) < \frac{2}{3}\OPT(E)$, the optimum contains two items that are more valuable than $\frac{1}{3} C$ and fit together and \nalgud\ returns a good solution. 
Before proceeding with the formal proof, we state some simple observations about the performance of \nalglarge~and its relation to \nalgud~with parameter $\beta=\lfrac{2}{3}$.

\begin{lemma}\label{lem:large-easy-cases}
    \nalglarge~is $\alpha$-approximate on $\calI_{\mathrm{UD}}$ for $\alpha=\max\big\{\frac{\max\{v_i \mid i\in E\}}{C}, \frac{1}{2} \big\}$, and for every set of items $E$ of an instance in $\calI_{\mathrm{UD}}$ it satisfies $v(\alglarge(E)) \geq \max\{v_i \mid i\in E\}$. Furthermore, for every set of items $E$ of an instance in $\calI_{\mathrm{UD}}$ such that $\max\{ v_i \mid i\in E\} \leq \frac{1}{2}C$ or \mbox{$\max\{ v_i \mid i\in E\} \geq \frac{2}{3}C$}, it holds that 
    \begin{equation*}
        \frac{v(\alglarge(E))}{v(\OPT(E))} = \frac{v(\algud_{2/3}(E))}{v(\OPT(E))} \geq \frac{2}{3}.
    \end{equation*}
\end{lemma}

\begin{proof}
    Let $E$ be an arbitrary set of items and denote $i_{\max} = \arg\max\{v_i \mid i\in E\}$ and \mbox{$v_{\max}=\max\{v_i \mid i\in E\}$}. If $v_{\max} \geq \frac{2}{3}C$, then $\max_{a \in [n]} v(\OPT(E_a)) \geq v_{\max} \geq \frac{2}{3} C$ and, denoting $a' = \arg\max_{a \in [n]} v(\OPT(E_a))$, we have that
    \begin{equation*}
        v(\alglarge(E)) = v(\algud_{2/3}(E)) = v(\OPT(E_{a'}) \geq \frac{2}{3}v(\OPT(E)),
    \end{equation*}
    due to $v(\OPT(E))\leq C$. The inequalities $v(\alglarge(E)) \geq v_{\max}$ and \mbox{$v(\alglarge(E)) \geq \frac{v_{\max}}{C}v(\OPT(E))$} follow directly in this case.
    Similarly, if $v_{\max} \leq \frac{1}{2}C$, then $R=S=E$, and therefore \mbox{$\alglarge(E)= \algud_{2/3}(E)=\algrspgreedy(E,E)$}. If $v(E)\leq C$, then \Cref{lem:r(s)-small} yields
    \begin{equation*}
        v(\algrspgreedy(E,E)) \geq v(E) = v(\OPT(E)).
    \end{equation*}
    If $v(E)>C$, then \Cref{lem:r(s)-large} yields
    \begin{equation*}
        v(\algrspgreedy(E,E)) \geq \max\Bigg\{ 1-\frac{v_{\max}}{C}, \frac{1}{2}+\frac{v_{\max}}{2C} \Bigg\} \geq  \frac{2}{3}v(\OPT(E)),
    \end{equation*}
    where the last inequality follows from computing the worst-case value of $v_{\max}/C$. This concludes the proof of the second claim in the statement. Note that, in these cases with $v_{\max} \geq \frac{2}{3}C$ or $v_{\max} \leq \frac{1}{2}C$, the fact that \nalglarge~is $\frac{1}{2}$-approximate follows directly. 

    It remains to prove that \nalglarge~is $\frac{1}{2}$-approximate when $\frac{1}{2}C<v_{\max}<\frac{2}{3}C$, that it is $\frac{v_{\max}}{C}$-approximate when $v_{\max}\leq \frac{2}{3}C$, and that $v(\alglarge(E)) \geq v_{\max}$ when $v_{\max}\leq \frac{2}{3}C$.
    
    We start with the last two claims.  
    If $v_{\max}\leq \frac{2}{3}C$, as $i_{\max}\in S$ is the item with maximum value we have that $x_{i_{\max}}(S)=1$ and thus
    \begin{align*}
        v(\alglarge(E)) & = v(\algrspgreedy(E,S)) \geq v_{a(i_{\max})}(\algrspgreedy(E,S)) \\
            & = \max\Bigg\{ v(E') \;\Bigg\vert\; E'\subseteq E_{a(i_{\max})} \text{ s.t.\ } v(E') \leq \sum_{i\in S_{a(i_{\max})}}v_ix_i(S) \Bigg\} \geq v_{\max}.
    \end{align*}
    This implies, in addition, that $v(\alglarge(E)) \geq \frac{v_{\max}}{C}v(\OPT(E))$, hence the $\frac{v_{\max}}{C}$-approximation. When $v_{\max}>\frac{1}{2}C$, this directly implies a $\frac{1}{2}$-approximation, which proves that claim for the missing case and concludes the proof.
\end{proof}

We have now all necessary ingredients to prove \Cref{theo:unit-density-rand-UB}.

\begin{proof}[Proof of \Cref{theo:unit-density-rand-UB}]
    We claim the result for \nalgmixud. Strategyproofness of this mechanism follows directly from the strategyproofness of both mechanisms over which it randomizes, which was already established in \Cref{lem:alg-ud-param} and \Cref{lem:sp-large}.
    
    In order to prove the approximation factor, we let $E$ be an arbitrary set of items and denote $i_{\max} = \arg\max\{v_i \mid i\in E\}$ and $v_{\max}=\max\{v_i \mid i\in E\}$. We first observe that whenever $E$ is such that both $v(\algud_{2/3}(E)) \geq \frac{2}{3} v(\OPT(E))$ and $v(\alglarge(E)) \geq \frac{2}{3} v(\OPT(E))$ hold, the result follows directly. We know from \Cref{lem:large-easy-cases} that this occurs, in particular, if $v_{\max} \leq \frac{1}{2}C$ or $v_{\max} \geq \frac{2}{3}C$, so we shall assume this is not the case. In the following, we address the cases with $v(\algud_{2/3}(E)) < \frac{2}{3} v(\OPT(E))$ and $v(\alglarge(E)) < \frac{2}{3} v(\OPT(E))$ separately, showing that the approximation factor still holds in both of them. 
    
    We first suppose that $v(\algud_{2/3}(E)) < \frac{2}{3} v(\OPT(E))$. \Cref{lem:alg-ud-param} implies that $v_{i^*} < \frac{2}{3}v_{\max}$, so in particular we know that $i^*\neq i_{\max}$.
    Since $v_{\max} < \frac{2}{3}C$, part \ref{item:i-star-not-fitting} of \Cref{lem:items-not-fitting} yields $v_{i^*} + v_{\max} > C$, thus
    \begin{equation}
        v_{i^*} > \frac{1}{3}C > \frac{1}{2}v_{\max}.\label{eq:lb-i-star}
    \end{equation}
    In addition, \Cref{lem:alg-ud-param} implies that $v(R)<\frac{2}{3}C$, hence $v(\algud_{2/3}(E)) = v(\algrspgreedy(E,R)) \geq v(R)$, where the inequality comes from \Cref{lem:r(s)-small}. As $v(\algud_{2/3}(E)) < \frac{2}{3} v(\OPT(E))$, we obtain that $v(R)<\frac{2}{3}C$. Plugging this into \eqref{eq:lb-i-star} yields
    \begin{equation}
        v(R\setminus \{i^*\}) < \frac{2}{3}C-\frac{1}{3}C = \frac{1}{3}C.\label{eq:ub-R-minus-i-star}
    \end{equation}
    On the other hand, we recall that from \Cref{lem:alg-ud-param} we know that $v_{i^*} < \frac{2}{3}v_{\max} < \frac{1}{2}C$, so every $i\in E$ with $v_i<v_{i^*}$ satisfies $v_{i^*}+v_i \leq C$. \Cref{lem:i-star-largest} then yields 
    \begin{equation}
        R=\{i \in E \mid v_i \leq v_{i^*}\}. \label{eq:def-R}
    \end{equation}
    Putting \eqref{eq:ub-R-minus-i-star} and \eqref{eq:def-R} together with the fact that $v_{\max}<\frac{2}{3}C$ and parts \ref{item:i-star-not-fitting} and \ref{item:deleted-items-not-fitting} of \Cref{lem:items-not-fitting}, we conclude the following:
    (i) for every $j,k\in E$ with $j\not=k$ and $\min\{v_j,v_k\} \geq v_{i^*}$, it holds that $v_j+v_k>C$; and (ii) for every $j\in E$, it holds that $v_j+v(R\setminus \{i^*\}) \leq C$. Therefore, $\OPT(E)=i_{\max}\cup (R\setminus \{i^*\})$.
    \Cref{lem:r(s)-small} and \Cref{lem:large-easy-cases} imply that $v(\algud_{2/3}(E)) \geq v(R)$ and $v(\alglarge(E)) \geq v_{\max}$, thus
    \begin{equation*}
        \frac{\E{v(\algmixud(E))}}{v(\OPT(E))} \geq \frac{\frac{2}{3}v(R)+\frac{1}{3}v_{\max}}{v_{\max}+ v(R) -v_{i^*}} \geq \frac{\frac{2}{3}v(R)+\frac{1}{3}v_{\max}}{v(R) + \frac{1}{2}v_{\max}} = \frac{2}{3},
    \end{equation*}
    where we used \eqref{eq:lb-i-star} in the last inequality once again.
    
    In the remainder of this proof, we suppose that $v(\alglarge(E)) < \frac{2}{3}v(\OPT(E))$. We further distinguish four cases.
    
    If $v(R)>C$ and $i_{\max}\in R$, using the fact that $v_{\max}>\frac{1}{2}C$, \Cref{lem:r(s)-large} implies that $v(\algud_{2/3}(E)) = v(\algrspgreedy(E,R)) \geq \frac{3}{4}v(\OPT(E))$, hence
    \begin{align*}
        \E{v(\algmixud(E))} = \frac{2}{3}v(\algud_{2/3}(E))+\frac{1}{3}v(\alglarge(E)) &\geq \frac{2}{3}\cdot\frac{3}{4}v(\OPT(E))+\frac{1}{3}\cdot\frac{1}{2}v(\OPT(E)) \\
        &= \frac{2}{3} v(\OPT(E)),
    \end{align*}
    where the inequality follows from the previous observation and \Cref{lem:large-easy-cases}.
    
    Similarly, if $v(R)>C$ and $i_{\max}\not\in R$, part \ref{item:i-star-not-fitting} of \Cref{lem:items-not-fitting} implies that $v_{i^*}>C-v_{\max}$. Therefore, \Cref{lem:r(s)-large} implies that $v(\algud_{2/3}(E)) = v(\algrspgreedy(E,R)) \geq \big(1-\frac{v_{\max}}{2C}\big)v(\OPT(E))$, hence
    \begin{align*}
        \E{v(\algmixud(E))} &= \frac{2}{3}v(\algud_{2/3}(E))+\frac{1}{3}v(\alglarge(E)) \\
        & \geq \frac{2}{3}\Big(1-\frac{v_{\max}}{2C}\Big)v(\OPT(E))+ \frac{v_{\max}}{3C}v(\OPT(E)) = \frac{2}{3}v(\OPT(E)),
    \end{align*}
    where the inequality follows from the previous observation and \Cref{lem:large-easy-cases}.
    
    We now consider the case $v(R) \leq C$, where we know that \mbox{$v(\algud_{2/3}(E)) = v(\algrspgreedy(E,R)) \geq v(R)$} due to \Cref{lem:r(s)-small}.
    
    If $i^* \in \OPT(E)$ we get that $\OPT(E) \subseteq R$, since we know by part~\ref{item:i-star-not-fitting} of \Cref{lem:items-not-fitting} that $v_{i^*} + v_j > C$ for all items $j \in E \setminus R$. Therefore, $v(\algud_{2/3}(E)) = v(\OPT(E))$ and we conclude that
    \begin{align*}
        \E{v(\algmixud(E))} &= \frac{2}{3}v(\algud_{2/3}(E))+\frac{1}{3}v(\alglarge(E)) \\
        &\geq \frac{2}{3}v(\OPT(E))+\frac{1}{3}\cdot \frac{1}{2}v(\OPT(E)) = \frac{5}{6} v(\OPT(E)),
    \end{align*}    
    where the inequality follows from the previous observation and \Cref{lem:large-easy-cases}.
    
    Finally, if $i^* \notin \OPT(E)$, we get $v(R) = v(R \setminus \OPT(E)) + v(R \cap \OPT(E)) \geq v_{i^*} + v(R \cap \OPT(E)) $. Part~\ref{item:deleted-items-not-fitting} of \Cref{lem:items-not-fitting} yields $|\OPT(E)\setminus R| \leq 1$, thus $v(\OPT(E)\setminus R) \leq v_{\max}$. Note that \mbox{$v_{i^*} \neq v_{\max}$} since otherwise, $\OPT(E) = R$ by part~\ref{item:i-star-not-fitting} of \Cref{lem:items-not-fitting} and $v(R)\leq C$, contradicting the hypothesis that $i^*\not\in \OPT(E)$.  Combining the last inequality with the fact that \mbox{$v(\alglarge(E)) < \frac{2}{3}v(\OPT(E))$}, we obtain
    \begin{equation*}
        v(R \cap \OPT(E)) = v(\OPT(E)) - v(\OPT(E) \setminus R) > \frac{3}{2} v(\alglarge(E)) - v_{\max}.
    \end{equation*}
    Since $v(\alglarge(E))\geq v_{\max}$ we get $v(R) > v_{i^*} + \frac{1}{2} v_{\max} $. We conclude that
    \begin{align*}
        \E{v(\algmixud(E))} & = \frac{2}{3}v(\algud_{2/3}(E))+\frac{1}{3}v(\alglarge(E)) \\
        &\geq \frac{2}{3} \Big(v_{i^*} + \frac{1}{2} v_{\max}\Big)+\frac{1}{3} v_{\max} = \frac{2}{3} \Big(v_{i^*} + v_{\max}\Big) > \frac{2}{3} C \geq \frac{2}{3} v(\OPT(E)),
    \end{align*}
    where we used part~\ref{item:i-star-not-fitting} of \Cref{lem:items-not-fitting} for the second to last inequality.
\end{proof}
    
We finish by providing a lower bound on the approximation factor that randomized mechanisms can achieve. It extends the lower bound by \citet{FeigenbaumJ17} for general mechanisms to the unit-density setting. 

\begin{theorem}
    \label{theo:unit-density-rand-LB}
    Let $f$ be a strategyproof and $\alpha$-approximate randomized mechanism on $\calI_{\mathrm{UD}}$. Then, $\alpha \leq 1/(5\phi-7) \approx 0.917$.
\end{theorem}

To prove this bound, we consider $n=2$ agents, sets of items $E_1=\{i,j\},~E'_1=\{i\}$, and $E_2=\{k\}$, values $v_i=\phi,~ v_j=1$, and $v_k=1$, and a capacity $C=2$. The bound then follows from a similar argument to \Cref{theo:unit-density-det-LB}.

\begin{proof}[Proof of \Cref{theo:unit-density-rand-LB}]
Suppose that $\scM$ is a strategyproof and \mbox{$\alpha$-approximate} mechanism for $\alpha> 1/(5\phi-7)$, say $\alpha=1/(5\phi-7-\varepsilon)$ for some $\varepsilon >0$.

We consider $n=2$ agents and sets of items $E_1=\{i,j\},~E'_1=\{i\}$, and $E_2=\{k\}$ with values $v_i=\phi,~ v_j=1$, and $v_k=1$. We fix a capacity  $C=2$ and consider two instances: $E=E_1\cup E_2$ and $E'=E'_1\cup E_2$. 

For the instance $E$, we have that $\OPT(E)=\{j,k\}$ and thus $v(\OPT(E))=2$. Letting $p=\P{\scM(E)=\{j,k\}}$ denote the probability of $\scM$ selecting this optimal set of items and noting that $v(D)\leq \phi$ for any $D\subseteq E$ with $D\not= \{j,k\}$ and $v(D)\leq C$, the fact that $\scM$ is $\frac{1}{5\phi-7-\varepsilon}$-approximate implies that
\begin{equation*}
    \frac{2p+\phi(1-p)}{2} \geq \frac{v(\scM(E))}{v(\OPT(E))} \geq \frac{1}{5\phi-7-\varepsilon} \Rightarrow p \geq \phi^2 \Big( \frac{2}{5\phi-7-\varepsilon} - \phi \Big).
\end{equation*}
Therefore,
\begin{equation}
    \E{v_1(\scM(E))} \leq p+\phi(1-p) \leq \phi \Big( \phi+1-\frac{2}{5\phi-7-\varepsilon}\Big).\label{eq:value-agent-1-E}
\end{equation}

On the other hand, for the instance $E'$, we have $\OPT(E')=\{i\}$ and thus $v(\OPT(E'))=\phi$. Letting $p'=\P{\scM(E')=\{i\}}$ denote the probability of $\scM$ selecting this optimal set of items and noting that $v(D)\leq 1$ for any $D\subseteq E'$ with $D\not= \{i\}$ and $v(D)\leq C$, the fact that $\scM$ is $1/(5\phi-7-\varepsilon)$-approximate implies that
\begin{equation*}
    \frac{\phi p'+1-p'}{\phi} \geq \frac{v(\scM(E'))}{v(\OPT(E')} \geq \frac{1}{5\phi-7-\varepsilon} \Rightarrow p' \geq \phi \Big( \frac{\phi}{5\phi-7-\varepsilon} - 1 \Big).
\end{equation*}
Therefore,
\begin{equation}
    \E{v_1(\scM(E'))} \geq \phi p' \geq \phi^2 \Big( \frac{\phi}{5\phi-7-\varepsilon} - 1 \Big).\label{eq:value-agent-1-E'}
\end{equation}

Since $\scM$ is strategyproof and $E'_1 \subseteq E_1$, it must hold that $\E{v_1(\scM(E))} \geq \E{v_1(\scM(E'))}$. Putting expressions \eqref{eq:value-agent-1-E} and \eqref{eq:value-agent-1-E'} together, this implies that
\begin{equation*}
    \phi \Big( \phi+1-\frac{2}{5\phi-7-\varepsilon}\Big) \geq \phi^2 \Big( \frac{\phi}{5\phi-7-\varepsilon} - 1 \Big) \Rightarrow \varepsilon \leq \frac{5(1+\phi-\phi^2)}{\phi(2\phi+1)} = 0,
\end{equation*}
a contradiction.
\end{proof}

\section{Discussion}

Our work contributes to the growing field of mechanism design without money by narrowing several gaps in the approximation factors achievable by strategyproof mechanisms in the context of knapsack constraints. We improve on previous results for the case of general constraints and introduce the natural setting of unit-density knapsack, where we obtain a deterministic mechanism providing the best-possible guarantee for this class of mechanisms and establish a strict separation between the factors attainable by deterministic and randomized mechanisms.

The fact that our mechanisms need to solve knapsack problems to optimality as subroutines poses an interesting direction for future research: The design of efficiently implementable mechanisms providing our (or better) approximation ratios up to arbitrary precision. In the case of randomized mechanisms, an alternative to \nalgrandspgreedy\ remains strategyproof (in expectation) and $\frac{1}{2}$-approximate and runs in polynomial time: Given items~$E$, return the integral greedy solution with probability~$\frac{1}{2}$ and the fractional item $i$ with probability \smash{$\frac{x_{i}(E)}{2}$}. In order to achieve a~\smash{$\big(\frac{1}{3}-\varepsilon\big)$}-approximation in the deterministic setting, instead of finding the optimum of the individual knapsack problems, we would need a PTAS for the classic knapsack problem with two natural monotonicity properties, requiring that deleting an item from the input or reducing the capacity cannot lead to an increase in the value of the set output by the PTAS.

From a broader perspective, one may wonder if greedy-based approaches like the ones exploited in this paper can provide better approximation factors in other contexts from previous literature, such as generalizations of knapsack and graph-theoretical problems. More generally, one can visualize this family of problems in terms of solving a linear problem with specific types of constraints, where agents control a subset of variables and may not benefit from hiding them from the mechanism designer. From this point of view, it is worth exploring under which structure of the linear program certain approaches provide non-trivial approximation factors.

\printbibliography

\end{document}